\documentclass{amsart}
\usepackage{amsmath,amsfonts,amssymb}
\usepackage{graphicx,amsmath,amscd}
\usepackage{color,comment}
\usepackage[T1]{fontenc}
\usepackage[latin1]{inputenc}
\usepackage{revsymb}

\usepackage[lmargin=2cm,rmargin=2cm,tmargin=2cm,bmargin=2cm]{geometry}
\usepackage[toc,page]{appendix}

\newtheorem{theorem}{Theorem}[section]
\newtheorem{lemma}[theorem]{Lemma}
\newtheorem{corollary}[theorem]{Corollary}
\newtheorem{proposition}[theorem]{Proposition}
\newtheorem{remark}[theorem]{Remark}
\newtheorem{defi/prop}[theorem]{Definition/Proposition}

\newcommand{\N}{\mathbf{N}}

\newcommand{\C}{\mathbf{C}}
\renewcommand{\P}{\mathbf{P}}

\renewcommand{\leq}{\leqslant}
\renewcommand{\geq}{\geqslant}

\newcommand{\st}{\  : \ }

\newcommand{\A}{\mathrm{A}}
\newcommand{\B}{\mathrm{B}}

\newcommand{\Id}{\openone}

\newcommand{\cH}{\mathcal{H}}

\newcommand{\cK}{\mathcal{K}}

\DeclareMathOperator{\Span}{Span}

\DeclareMathOperator{\Tr}{Tr}
\DeclareMathOperator{\Sym}{Sym}

\newcommand{\braket}[2]{\langle #1 | #2\rangle}

\newcommand{\proj}[1]{| #1 \rangle\!\langle #1 |}
\newcommand{\bra}[1]{\langle #1 |}
\newcommand{\ket}[1]{| #1 \rangle}

\begin{document}

\title{Flexible constrained de Finetti reductions and applications}

\date{May 29th 2016}

\author{C\'{e}cilia Lancien}

\author{Andreas Winter}

\address{\textbf{C\'{e}cilia Lancien:} Institut Camille Jordan, Universit\'{e} Claude Bernard Lyon 1, 69622 Villeurbanne Cedex, France
\& Departament de F\'{\i}sica: Grup d'Informaci\'{o} Qu\`{a}ntica, Universitat Aut\`{o}noma de Barcelona, 08193 Bellaterra, Barcelona, Spain.}
\email{lancien@math.univ-lyon1.fr}

\address{\textbf{Andreas Winter:} Departament de F\'{\i}sica: Grup d'Informaci\'{o} Qu\`{a}ntica, Universitat Aut\`{o}noma de Barcelona, 08193 Bellaterra, Barcelona, Spain \& Instituci\'{o} Catalana de Recerca i Estudis Avan\c{c}ats (ICREA), 08010 Barcelona, Spain.}
\email{andreas.winter@uab.cat}

\begin{abstract}
De Finetti theorems show how sufficiently exchangeable states are well-approximated by convex combinations of i.i.d.~states. Recently,
it was shown that in many quantum information applications a more relaxed \emph{de Finetti reduction} (i.e.~only a matrix inequality
between the symmetric state and one of de Finetti form) is enough, and that it leads to more concise and elegant arguments.

Here we show several uses and general flexible applicability of a \emph{constrained de Finetti reduction} in quantum information theory, which was recently discovered by Duan, Severini and Winter. In particular we show that the technique can accommodate other symmetries commuting with the permutation action, and permutation-invariant linear constraints.
We then demonstrate that, in some cases, it is also fruitful with convex constraints, in particular separability in a bipartite setting. This is a constraint particularly interesting in the context of the complexity class $\mathrm{QMA}(2)$ of interactive quantum Merlin-Arthur games with unentangled provers, and our results relate to the soundness gap amplification of $\mathrm{QMA}(2)$ protocols by parallel repetition. It is also relevant for the regularization of certain entropic channel parameters. Finally, we explore an extension to infinite-dimensional systems, which usually pose inherent problems to de Finetti techniques in the quantum case.
\end{abstract}

\maketitle

%%%%%% Auf geht's beim Schichtl... %%%%%%

\section{Introduction}
\label{sec:intro}
The main motivation behind all de Finetti type theorems is to reduce
the study of permutation-invariant scenarios to that of i.i.d.~ones,
which are often much easier to understand.
In many information theoretic situations, the problem is posed in such
a way that one almost directly sees that the solution is (or is without
loss of generality) permutation-invariant. Furthermore, in many
scenarios one needs only to upper bound (and not to accurately approximate)
a permutation-invariant object by i.i.d.~ones. The seminal \textit{de Finetti
reduction} (aka \textit{post-selection lemma}) of Christandl, K\"{o}nig and Renner \cite{CKR}
was precisely designed for that: for any permutation-invariant
state $\rho$ on $\mathcal{H}^{\otimes n}$, with $d=|\mathcal{H}|$ the
``local'' Hilbert space dimension,
\begin{equation}
  \label{eq:CKR}
  \rho \leq (n+1)^{d^2} \int_{\sigma\in\mathcal{D}(\cH)} \sigma^{\otimes n} \,\mathrm{d}\sigma,
\end{equation}
where $\mathrm{d}\sigma$ is a universal probability measure over the set of mixed states $\mathcal{D}(\cH)$ on $\mathcal{H}$,
and the inequality refers to the matrix order ($A\leq B$ meaning
that $B-A$ is positive semidefinite).
The beauty of this statement is that on the right hand side we have a
universal object: one and the same convex combination provides the
upper bound to all permutation-invariant states.
At the same time, though, its very universality can be a drawback:
every permutation-invariant state (quantum or classical) is upper bounded
by the same convex combination of tensor power states, so that any
other a priori information (apart from its permutation-symmetry),
that one may have on it, is lost. In~\cite[Appendix~B]{DSW}, it was shown that
at the sole cost of slightly increasing the polynomial pre-factor in
front of the upper bounding de Finetti operator, it is actually possible
to make it depend on the state of interest, or on some property
that this state has, including in the integral on the right hand side
of equation~(\ref{eq:CKR}) a fidelity term between $\rho$ and the
i.i.d.~state $\sigma^{\otimes n}$.
In~\cite{DSW}, this \emph{constrained de Finetti reduction}
was applied to prove a coding
theorem in a setting with adversarially chosen channel.
In~\cite{LaW2} another application to parallel repetition
of no-signalling games was given.
%\textcolor{red}{Mention approximate microcanonical subspace construction [Cf. arXiv:1512.01189]?}

In Section~\ref{sec:finite-de-finetti}, we first review the constrained de
Finetti reduction of~\cite[Appendix~B]{DSW}, for the sake both of completeness and of presenting its proof in a slightly alternative way
(Subsection~\ref{subsec:finite-general}). We then show that certain \emph{linear} constraints lead to very simple
and at the same time useful forms of the de Finetti reduction, such
that certain ``unwanted'' contributions in the integral on the right hand
side of equation~(\ref{eq:CKR}) are either completely absent or exponentially
suppressed (Subsection~\ref{subsec:linear-constraints}).
Next, in Sections~\ref{sec:sep1} and \ref{sec:sep2} we study in depth
the case of separability, a convex constraint. In particular we show
that there are several essentially equivalent ways of thinking about
the exponential decay of the fidelity term.
Inspired by separability, in Section~\ref{sec:general} we present an
axiomatic treatment of a wider class of convex constraints.
Finally, in Section~\ref{sec:infinite} we move to de Finetti reductions
in the infinite-dimensional case.

\section{Flexible de Finetti reductions for finite-dimensional symmetric quantum systems}
\label{sec:finite-de-finetti}

\subsection{A general constrained de Finetti reduction}
\label{subsec:finite-general}
Before getting into more specific statements, let us fix once and for
all some definitions and notation that we shall use throughout the
whole paper. Consider $\cH$ a finite-dimensional Hilbert space, and denote by $\{\ket{1},\ldots,\ket{d}\}$ an orthonormal basis of $\cH$, where $d=|\mathcal{H}|<+\infty$.
Next, for any natural number $n$ and permutation $\pi\in\mathcal{S}_n$,
define $U_{\pi}$ as the associated permutation unitary on
$\mathcal{H}^{\otimes n}$, characterized by
\[
  \forall\ 1\leq j_1,\ldots, j_n\leq d,\ U_{\pi}\ket{j_{1}}\otimes\cdots\otimes\ket{j_{n}}
           = \ket{j_{\pi(1)}}\otimes\cdots\otimes\ket{j_{\pi(n)}}.
\]
Note that this definition is independent of the basis.
The $n$-symmetric subspace of $\mathcal{H}^{\otimes n}$ can then be defined as
the simultaneous $+1$-eigenspace of all $U_\pi$'s,
\begin{align*}
  \Sym^n(\mathcal{H})
     :=& \left\{ \ket{\psi}\in\cH^{\otimes n}
                \st \forall\ \pi\in\mathcal{S}_n,\ U_{\pi}\ket{\psi}=\ket{\psi} \right\} \\
      =&  \Span\left\{ \ket{v_{j_1,\ldots,j_n}}=\sum_{\pi\in\mathcal{S}_n}
                \ket{j_{\pi(1)}}\otimes\cdots\otimes\ket{j_{\pi(n)}} \st 1\leq j_1\leq\cdots\leq j_n\leq d \right\}.
\end{align*}
The orthogonal projector onto $\Sym^n(\mathcal{H})$ may thus be written as
\[
  P_{\Sym^n(\mathcal{H})}
              = \sum_{1\leq j_1\leq\cdots\leq j_n\leq d} \proj{\psi_{j_1,\ldots,j_n}}
              = {n+d-1 \choose n} \int_{\ket{\psi}\in S_{\cH}} \proj{\psi}^{\otimes n}\mathrm{d}\psi,
\]
where for each $1\leq j_1\leq\cdots\leq j_n\leq d$, $\ket{\psi_{j_1,\ldots,j_n}}$
denotes the unit vector having the same direction as $\ket{v_{j_1,\ldots,j_n}}$,
and where $\mathrm{d}\psi$ stands for the uniform probability measure
on the unit sphere $S_{\cH}$ of $\mathcal{H}$. The second line is due to Schur's Lemma, since $\Sym^n\left(\cH\right)$ is an irreducible representation (irrep) of the commutant action of $\{U_\pi:\pi\in S_n\}$,
the local unitaries $V^{\otimes n}$, $V\in SU(\mathcal{H})$ (see e.g.~\cite{Harrow} for more details).

A state $\rho$ on $\mathcal{H}^{\otimes n}$ is then called permutation-invariant (or simply symmetric) if
$U_{\pi}\rho U_{\pi}^{\dagger} = \rho$ for all $\pi\in\mathcal{S}_n$.
This can be expressed equivalently by saying that there exists a
unit vector $\ket{\psi}\in\Sym^n(\mathcal{H}\otimes\mathcal{H}')$ such that
$\rho = \Tr_{\mathcal{H}'^{\otimes n}} \proj{\psi}$.

Going from rigid to more flexible de Finetti reductions relies essentially on
the so-called ``pinching trick'', which we state formally as Lemma \ref{lemma:pinching} below. This is a generalization of results appearing in \cite{Hayashi} and \cite{HO}.

\begin{lemma}
\label{lemma:pinching}
Let $\cH$ be a Hilbert space and $M_1,\ldots,M_r$ be operators on $\cH$.
Then, for any state $\rho$ on $\cH$,
\[
  \sum_{i,j=1}^r M_i\rho M_j^{\dagger} \leq r\sum_{i=1}^r M_i\rho M_i^{\dagger}.
\]
\end{lemma}

\begin{proof}
To prove that Lemma \ref{lemma:pinching} holds for any state on $\cH$, it is
sufficient to prove that it holds for any pure state on $\cH$. Let therefore
$\ket{\psi}$ be a unit vector in $\cH$. Then, for any unit vector $\ket{\varphi}$ in $\cH$,
we have by the Cauchy-Schwarz inequality
\[\begin{split}
  \bra{\varphi} \left(\sum_{i,j=1}^r M_i\proj{\psi} M_j^{\dagger}\right) \ket{\varphi}
     &=    \left| \sum_{i=1}^r \bra{\varphi}M_i\ket{\psi} \right|^2 \\
     &\leq r \sum_{i=1}^r \left|\bra{\varphi}M_i\ket{\psi} \right|^2 \\
     &=    \bra{\varphi} \left( r\sum_{i=1}^r M_i\proj{\psi} M_i^{\dagger}\right) \ket{\varphi},
\end{split}\]
which concludes the proof.
\end{proof}

With this tool at hand, we are ready to get, first of all,
the pure state version of the flexible de Finetti reduction.

\begin{proposition}
  \label{prop:ps-pure}
  Any unit vector $\ket{\theta}\in\Sym^n\left(\cH\right)$ satisfies
  \[
     \proj{\theta} \leq {n+d-1 \choose n}^3 \int_{\ket{\psi}\in S_{\cH}}
             \left|\langle\theta|\psi^{\otimes n}\rangle\right|^2 \proj{\psi}^{\otimes n}\mathrm{d}\psi.
  \]
\end{proposition}
\begin{proof}
Let $\ket{\theta}\in\Sym^n\left(\cH\right)$ be a unit vector. Then,
\[
  \proj{\theta} = P_{\Sym^n(\cH)} \proj{\theta} P_{\Sym^n(\cH)}^{\dagger}
                = {n+d-1 \choose n}^2 \int_{\ket{\psi},\ket{\varphi}\in S_{\cH}}
                                \proj{\psi}^{\otimes n} \proj{\theta}
                                \proj{\varphi}^{\otimes n} \,\mathrm{d}\psi\,\mathrm{d}\varphi.
\]
Now observe, setting $r={n+d-1 \choose n}^2$, that the span of $\left\{ \proj{\psi}^{\otimes n},\ \ket{\psi}\in S_{\cH} \right\}$, subject to the condition of having trace $1$, has dimension $r-1$. So by Caratheodory's theorem, we know that there exist $\{p_1,\ldots,p_r\}$, a convex combination, and $\{\psi_1,\ldots,\psi_r\}$, a set of unit vectors in $\cH$, such that
\begin{equation} \label{eq:cara} \int_{\ket{\psi}\in S_{\cH}} \proj{\psi}^{\otimes n}\mathrm{d}\psi = \sum_{i=1}^r p_i \proj{\psi_i}^{\otimes n}. \end{equation}
%, for any function $f:S_{\mathrm{H}}\rightarrow\R$,
%\[ \int_{\ket{\psi}\in S_{\mathrm{H}}} f(\psi) \proj{\psi}^{\otimes n}\mathrm{d}\psi = \sum_{i=1}^r p_i f(\psi_i)\proj{\psi_i}^{\otimes n}. \]
We can therefore rewrite
\[ \proj{\theta} = r \sum_{i,j=1}^r p_ip_j \proj{\psi_i}^{\otimes n}\proj{\theta}\proj{\psi_j}^{\otimes n} \leq r^2 \sum_{i=1}^r p_i^2\left|\braket{\theta}{\psi_i^{\otimes n}}\right|^2 \proj{\psi_i}^{\otimes n} \leq r^{3/2} \sum_{i=1}^r p_i\left|\braket{\theta}{\psi_i^{\otimes n}}\right|^2 \proj{\psi_i}^{\otimes n}, \]
where the next to last inequality is by Lemma \ref{lemma:pinching}, and the last inequality is because, for each $1\leq i\leq r$, $p_i\leq 1/\sqrt{r}$ (which can be seen by contracting both sides of equation \eqref{eq:cara} with $\bra{\psi_i^{\otimes n}}\cdot\ket{\psi_i^{\otimes n}}$). And consequently, since this holds for any ensemble $\{p_i,\,\psi_i\}_{1\leq i\leq r}$ satisfying equation \eqref{eq:cara}, we have by convex combination
\[ \proj{\theta} \leq r^{3/2} \int_{\ket{\psi}\in S_{\cH}} \left|\langle\theta|\psi^{\otimes n}\rangle\right|^2 \proj{\psi}^{\otimes n} \,\mathrm{d}\psi, \]
%And since by assumption on the ensemble $\{p_i,\,\psi_i\}_{1\leq i\leq r}$,
%\[ \sum_{i=1}^r p_i\left|\braket{\theta}{\psi_i^{\otimes n}}\right|^2 \proj{\psi_i}^{\otimes n} = \int_{\ket{\psi}\in S_{\mathrm{H}}} \left|\langle\theta|\psi^{\otimes n}\rangle\right|^2 \proj{\psi}^{\otimes n} \,\mathrm{d}\psi, \]
which is precisely the advertised result.
\end{proof}

From Proposition \ref{prop:ps-pure}, we can now easily derive the general
mixed state version of our flexible de Finetti reduction, which was originally obtained in \cite{DSW} by a slightly different route.

\begin{theorem}[Cf.~{\cite[Lemma 18]{DSW}}]
  \label{th:ps-mixed}
  Any symmetric state $\rho$ on $\cH^{\otimes n}$ satisfies
  \[
     \rho \leq {n+d^2-1 \choose n}^3 \int_{\ket{\psi}\in S_{\cH\otimes\cH'}}
        F\left(\rho,\sigma(\psi)^{\otimes n}\right)^2 \sigma(\psi)^{\otimes n} \,\mathrm{d}\psi,
  \]
  where for a unit vector $\ket{\psi}\in\cH\otimes\cH'$,
  $\sigma(\psi) = \Tr_{\cH'}\proj{\psi}$ is the reduced state of $\proj{\psi}$ on $\cH$.
\end{theorem}
\begin{proof}
As noted before, there exists a unit vector $\ket{\theta}\in\Sym^n\left(\cH\otimes\cH'\right)$
such that $\rho=\Tr_{\cH'^{\otimes n}}\proj{\theta}$.
By Proposition \ref{prop:ps-pure}, we have
\[
  \proj{\theta} \leq {n+d^2-1 \choose n}^3 \int_{\ket{\psi}\in S_{\cH\otimes\cH'}}
       \left|\langle\theta|\psi^{\otimes n}\rangle\right|^2
       \proj{\psi}^{\otimes n} \,\mathrm{d}\psi.
\]
Thus, after partial tracing over $\cH'^{\otimes n}$, we obtain
\[
  \rho\leq {n+d^2-1 \choose n}^3 \int_{\ket{\psi}\in S_{\cH\otimes\cH'}}
       \left|\langle\theta|\psi^{\otimes n}\rangle\right|^2
       \sigma(\psi)^{\otimes n} \,\mathrm{d}\psi.
\]
To get the announced result, we then just have to notice that,
by monotonicity of the fidelity under the CPTP map $\Tr_{\cH'^{\otimes n}}$,
we have for each $\ket{\psi}\in\cH\otimes\cH'$,
\[
  \left|\langle\theta|\psi^{\otimes n}\rangle\right|
       =    F\left(\proj{\theta}, \proj{\psi}^{\otimes n} \right)
       \leq F\left(\rho,\sigma(\psi)^{\otimes n}\right). \qedhere
\]
\end{proof}

\subsection{Linear constraints}
\label{subsec:linear-constraints}
Let $\rho$ be a symmetric state on $\cH^{\otimes n}$.
What Theorem \ref{th:ps-mixed} tells us is that there exists a
probability measure $\mu$ over the set of states on $\cH$ such that
\begin{equation}
  \label{eq:ps-fidelity}
   \rho \leq (n+1)^{3d^2} \int_{\sigma\in\in\mathcal{D}(\cH)}
                            F\left(\rho,\sigma^{\otimes n}\right)^2
                            \sigma^{\otimes n} \,\mathrm{d}\mu(\sigma).
\end{equation}
It may be pointed out that $\mu$ is in fact the uniform probability measure over the set of mixed states on $\mathcal{H}$ (with respect to the Hilbert--Schmidt distance), since the latter is equivalently characterized as the partial trace over an environment $\mathcal{H}'$ having same dimension as $\mathrm{H}$ of uniformly distributed pure states on $\mathcal{H}\otimes\mathcal{H}'$ (see \cite{ZS}).

Observe that, contrary to the original de Finetti reduction, where the upper bound is the same for every symmetric state, we here have a highly
state-dependent upper bound, where only states which have a high
fidelity with the state of interest $\rho$ are given an important weight.
This is especially useful when one knows that $\rho$ satisfies some
additional property. Indeed, one would then expect that, amongst states
of the form $\sigma^{\otimes n}$, only those approximately satisfying
this same property should have a non-negligible fidelity weight. There
are at least two archetypical cases where this intuition can easily be seen to be true.

\begin{corollary}[Cf.~{\cite[Lemma 18]{DSW}}]
  \label{cor:product-image}
  Let $\mathcal{N}:\mathcal{L}(\cH)\rightarrow\mathcal{L}(\mathcal{K})$ be a quantum channel, with $d=|\cH|<+\infty$.
  Assume that $\rho$ is a symmetric state on $\cH^{\otimes n}$, which
  is additionally satisfying $\mathcal{N}^{\otimes n}(\rho)=\tau_0^{\otimes n}$,
  for some given state $\tau_0$ on $\mathcal{K}$. Then,
  \[
     \rho \leq (n+1)^{3d^2} \int_{\sigma\in\mathcal{D}(\cH)}
             F\left(\tau_0,\mathcal{N}(\sigma)\right)^{2n} \sigma^{\otimes n}\,\mathrm{d}\mu(\sigma).
  \]
\end{corollary}
\begin{proof}
This follows directly from inequality \eqref{eq:ps-fidelity}
by monotonicity of the fidelity under the CPTP map $\mathcal{N}$,
and by multiplicativity of the fidelity on tensor products.
\end{proof}

This especially implies that, under the hypotheses of
Corollary \ref{cor:product-image}, we have: for any $0<\delta<1$,
setting
$\mathcal{K}_{\delta}=\left\{ \sigma\in\mathcal{D}(\cH) \st F\left(\tau_0,\mathcal{N}(\sigma)\right)\geq 1-\delta\right\}$,
\[
   \rho \leq (n+1)^{3d^2} \left( \int_{\sigma\in\mathcal{K}_{\delta}} \sigma^{\otimes n}\mathrm{d}\mu(\sigma)
        + (1-\delta)^{2n}\int_{\sigma\notin\mathcal{K}_{\delta}} \sigma^{\otimes n}\mathrm{d}\mu(\sigma) \right).
\]
Such flexible de Finetti reduction, for states which satisfy the constraint of
being sent to a certain tensor power state by a certain tensor power
CPTP map, has already been fruitfully applied, for instance in the context of
zero-error communication via quantum channel~\cite{DSW}.

\medskip
Another linear constraint is that of a fixed point equation:
\begin{corollary}
  \label{cor:fixed-point}
  Let $\mathcal{N}:\mathcal{L}(\cH)\rightarrow\mathcal{L}(\cH)$ be a quantum channel, with $d=|\cH|<+\infty$.
  Assume that $\rho$ is a symmetric state on $\cH^{\otimes n}$, which
  is additionally satisfying $\mathcal{N}^{\otimes n}(\rho)=\rho$. Then,
  \[
     \rho \leq (n+1)^{3d^2} \int_{\sigma\in\mathcal{D}(\cH)}
         F\left(\rho,\mathcal{N}(\sigma)^{\otimes n}\right)^2 \mathcal{N}(\sigma)^{\otimes n}\mathrm{d}\mu(\sigma).
  \]
\end{corollary}

\begin{proof}
Apply $\mathcal{N}^{\otimes n}$ on both sides of inequality \eqref{eq:ps-fidelity},
and use once more the monotonicity of the fidelity under the CPTP map $\mathcal{N}$.
\end{proof}

This means that, under the assumptions of Corollary \ref{cor:fixed-point}, there actually exists a probability measure $\widetilde{\mu}$ over the set of states on $\cH$ which belong to the range of $\mathcal{N}$ such that
\begin{equation}
  \label{eq:fixed-point}
  \rho \leq (n+1)^{3d^2} \int_{\sigma\in\mathrm{Range}(\mathcal{N})}
      F\left(\rho,\sigma^{\otimes n}\right)^2 \sigma^{\otimes n}\mathrm{d}\widetilde{\mu}(\sigma).
\end{equation}

A case of particular interest for equation \eqref{eq:fixed-point} is the following. Let $G$ be a subgroup of the unitary group on $\cH$, equipped with its Haar measure $\mu_G$ (unique normalised left and right invariant measure over $G$). Its associated twirl is the quantum channel $\mathcal{T}_G:\mathcal{L}(\cH)\rightarrow\mathcal{L}(\cH)$ defined by
\[ \mathcal{T}_G:\sigma\mapsto\int_{U\in G}U\sigma U^{\dagger}\mathrm{d}\mu_G(U). \]
The range of $\mathcal{T}_G$ is then precisely the set of states on $\cH$ in the commutant of $G$, i.e.~\[ \mathcal{K}_G=\left\{\sigma\in\mathcal{D}(\cH) \st \forall\ U\in G,\ [\sigma,U]=0 \right\}. \]
Hence, there exists a probability measure $\widetilde{\mu}$ over $\mathcal{K}_G$ such that, if $\rho$ is a symmetric state on $\cH^{\otimes n}$ satisfying $\mathcal{T}_G^{\otimes n}(\rho)=\rho$, then
\[
  \rho \leq (n+1)^{3d^2} \int_{\sigma\in\mathcal{K}_G}
         F\left(\rho,\sigma^{\otimes n}\right)^2 \sigma^{\otimes n}\mathrm{d}\widetilde{\mu}(\sigma).
\]

Another situation where equation \eqref{eq:fixed-point} might be especially useful is when $\mathcal{N}$ is a quantum-classical channel, so that its range can be identified with the set of classical probability distributions. We get in that case the corollary below.
\begin{corollary}
  \label{cor:probability}
  Let $\mathcal{X}$ be a finite alphabet and let $P_{\mathcal{X}^n}$ be a
  symmetric probability distribution on $\mathcal{X}^n$. There exists
  a universal probability measure $\mathrm{d}Q_{\mathcal{X}}$ over the set of
  probability distributions on $\mathcal{X}$ such that
  \[
    P_{\mathcal{X}^n} \leq (n+1)^{3|\mathcal{X}|^2} \int_{Q_{\mathcal{X}}}
              F\left(P_{\mathcal{X}^n},Q_{\mathcal{X}}^{\otimes n}\right)^2
              Q_{\mathcal{X}}^{\otimes n}\mathrm{d}Q_{\mathcal{X}},
  \]
  where the inequality sign signifies point-wise inequality between probability
  distributions on $\mathcal{X}^n$.
\end{corollary}

\begin{proof}
This is a special case of Corollary \ref{cor:fixed-point}.
Indeed, we know that we can make the identification $\mathcal{X}\equiv\{1,\ldots,d\}$,
where $d=|\mathcal{X}|$.  So let $\cH$ be a $d$-dimensional Hilbert space, and denote by $\{\ket{1},\ldots,\ket{d}\}$ an orthonormal basis of $\cH$. We can then define the ``classical'' state $\rho$ on $\cH^{\otimes n}$ by
\[
  \rho= \sum_{1\leq x_1,\ldots,x_n\leq d} P(x_1,\ldots,x_n)\proj{x_1\otimes\cdots\otimes x_n},
\]
and the quantum-classical channel $\mathcal{N}:\mathcal{L}(\cH)\rightarrow\mathcal{L}(\cH)$ by
\[
  \mathcal{N}:\sigma \mapsto \sum_{1\leq x\leq d} Q_{\sigma}(x)\proj{x}
                             = \sum_{1\leq x\leq d} \proj{x}\sigma\proj{x}.
\]
By assumption on $P$, $\rho$ is a symmetric state on $\cH^{\otimes n}$, which is additionally, by construction, a fixed point of $\mathcal{N}^{\otimes n}$. Hence, by Corollary \ref{cor:fixed-point},
\[ \rho \leq (n+1)^{3d^2} \int_{\sigma\in\mathcal{D}(\cH)} F\left(\rho,\mathcal{N}(\sigma)^{\otimes n}\right)^2 \mathcal{N}(\sigma)^{\otimes n}\mathrm{d}\mu(\sigma). \]
By the way $\rho$ and $\mathcal{N}$ have been designed, this actually translates into the point-wise inequality
\[ \forall\ 1\leq x_1,\ldots,x_n\leq d,\ P(x_1,\ldots,x_n) \leq (n+1)^{3d^2} \int_{\sigma\in\mathcal{D}(\cH)} F\left(P,Q_{\sigma}^{\otimes n}\right)^2 Q_{\sigma}(x_1)\cdots Q_{\sigma}(x_n)\mathrm{d}\mu\left(Q_{\sigma}\right), \]
which is exactly the announced result.
\end{proof}

This flexible de Finetti reduction for probability distributions turns out
to be especially useful when studying the parallel repetition of multi-player
non-local games, as exemplified in \cite{LaW2}.

\begin{remark}
Note that all these results generalize to non-normalized permutation
invariant positive semidefinite operators on finite-dimensional spaces (or positive
distributions on finite alphabets). One just has to extend the usual
definition of the fidelity by setting $F(M,N)=\|\sqrt{M}\sqrt{N}\|_1$
for any positive semidefinite operators (or positive distributions) $M,N$.
\end{remark}

\subsection{On to convex constraints?}
\label{subsec:convex-constraints}
We just saw that, in the case where the symmetric state $\rho$ under
consideration is additionally known to satisfy certain linear constraints,
it is possible to upper bound it by a de Finetti operator where
either no or exponentially small weight is given to tensor power
states which do not satisfy this same constraint. But what about
the case where the a priori information on $\rho$ is that it
belongs or not to a certain convex subset of states? This is the
question we investigate in the sequel, focussing first in
Sections \ref{sec:sep1} and \ref{sec:sep2} on the paradigmatic example
of the set of separable states, and then describing in
Section \ref{sec:general} the general setting in which similar conclusions hold.

\section{Exponential decay and concentration of $h_{sep}$ via de Finetti reduction approach}
\label{sec:sep1}

As we just mentioned, we will now be interested for a while in the case where the underlying Hilbert space is a tensor product Hilbert space $\cH = \mathrm{A}\otimes \mathrm{B}$, and the kind of symmetric states on $\cH^{\otimes n}$ that we will look at are those which additionally satisfy the (convex but non-linear) constraint of being separable across the bipartite cut $\A^{\otimes n}{:}\B^{\otimes n}$. For such a state $\rho$, one can of course still write down a de Finetti reduction of the form
\[ \rho \leq (n+1)^{3|\A|^2|\B|^2} \int_{\sigma\in\mathcal{D}(\mathrm{A}\otimes\mathrm{B})}
                     F\left(\rho,\sigma^{\otimes n}\right)^2 \sigma^{\otimes n}\,\mathrm{d}\mu(\sigma). \]
And what we would like to understand is whether it is possible to argue that only the states $\sigma^{\otimes n}$ which are such that $\sigma$ is separable across the bipartite cut $\A{:}\B$ are given a non exponentially small weight in this integral representation. As we shall see, this question is especially relevant when analysing the multiplicative behaviour of the support function of the set of biseparable states.

So let us specify a bit what we have in mind. Given a positive operator $M$ on $\mathrm{A}\otimes \mathrm{B}$, its maximum overlap with states which are
separable across the bipartite cut $\mathrm{A}{:}\mathrm{B}$, which we denote by $\mathcal{S}(\mathrm{A}{:}\mathrm{B})$, is defined as
\[
  h_{sep}(M) = \sup_{\sigma\in\mathcal{S}(\mathrm{A}{:}\mathrm{B})} \Tr(M\sigma).
\]
Here, we are interested in understanding how this quantity behaves
under tensoring. Concretely, this means that we want to know,
for any $n\in\N$, how $h_{sep}(M^{\otimes n})$ relates to $h_{sep}(M)$
(where the former quantity is defined as the maximum overlap of $M^{\otimes n}$ with states which are separable across the bipartite cut $\mathrm{A}^{\otimes n}{:}\mathrm{B}^{\otimes n}$). Because $h_{sep}$ is linear homogeneous in its argument, we can always rescale $M$ by a positive constant such that $0\leq M \leq \Id$, meaning that $M$ can be interpreted as a POVM element of the binary test with operators $(M,\Id-M)$. We shall make this assumption throughout from now on.
Then, it is easy to see that, for any $n\in\N$, we have the inequalities
\begin{equation}
  \label{eq:h_SEP-h_SEPn}
  h_{sep}(M)^n \leq h_{sep}(M^{\otimes n}) \leq h_{sep}(M) \leq 1.
\end{equation}
But in the case where $h_{sep}(M)<1$, the gap between the lower and upper bounds in
equation \eqref{eq:h_SEP-h_SEPn} grows exponentially with $n$, making these inequalities very little
informative.

This problem is interesting in itself, but also because it connects to plethora of others, some of them even outside the purely quantum information range of applications. The reader is referred to \cite{HM} for a full list of problems which are exactly or approximately equivalent to estimating $h_{sep}$. Two notable applications of $h_{sep}$ arise in quantum computing and in quantum Shannon theory: The first is to $\mathrm{QMA}(2)$, the class of quantum Merlin-Arthur interactive proof systems with two unentangled provers. The setting is that a verifier requires states $\alpha$ and $\beta$ from separate provers which are assumed to be computationally unlimited, and then performs a binary test with POVM $(M,\Id-M)$ on the separable state $\alpha\otimes\beta$. The maximum probability of passing the test that the provers can achieve, evidently equals precisely $h_{sep}(M)$. For complexity theoretic considerations (in particular the so-called soundness gap amplification) it is important to understand how well many instances of the same test, performed in parallel, can be passed -- either all $n$, leading to $h_{sep}(M^{\otimes n})$, or $t$ out of $n$, where $t > n h_{sep}(M)$. The second application appears in the problem of minimum output entropies of quantum channels, and their asymptotic behaviour. Namely, a quantum channel $\mathcal{N} : \mathcal{L}(\mathrm{A}) \rightarrow\mathcal{L}(\mathrm{B})$ can be represented in Stinespring form $\mathcal{N}(\rho) = \Tr_{\mathrm{E}} (V \rho V^{\dagger})$, with an isometry $V:\mathrm{A}\hookrightarrow \mathrm{B} \otimes \mathrm{E}$. Its minimum output R\'{e}nyi $p$-entropy is given by
\[ \widehat{S}_p(\mathcal{N}) = \min_{\rho\in\mathcal{D}(\mathrm{B})} S_p(\mathcal{N}(\rho)),\ \text{where}\ \forall\ \sigma\in\mathcal{D}(\mathrm{A}),\ S_p(\sigma) = \frac{1}{1-p}\log \Tr\sigma^p. \]
For $p=1$, taking the limit, we recover the von Neumann entropy, while for $p=\infty$, $S_{\infty}(\sigma) = -\log\| \sigma \|_\infty$. From this, it is not hard to see that, with $M = VV^{\dagger}$ the projector onto the range of $V$, i.e.~the subspace $V(\mathrm{A}) \subset \mathrm{B}\otimes \mathrm{E}$, we have $\widehat{S}_{\infty}(\mathcal{N}) = -\log h_{sep}(M)$. In quantum Shannon theory, the asymptotic behaviour of $\widehat{S}_p(\mathcal{N}^{\otimes n})$ is of great interest.

\subsection{Some general facts about ``filtered by measurements'' distance measures}

We need to introduce first a few definitions and properties regarding ``filtered by measurements'' distance measures.

Let $\cH$ be a Hilbert space and let $\mathbf{M}$ be a set of POVMs on $\cH$. For any states $\rho,\sigma$ on $\cH$, we define their measured by $\mathbf{M}$ trace-norm distance as
\[ D_{\mathbf{M}}(\rho,\sigma)=\sup_{\mathcal{M}\in\mathbf{M}}\frac{1}{2}\left\|\mathcal{M}(\rho)-\mathcal{M}(\sigma)\right\|_1, \]
and their measured by $\mathbf{M}$ fidelity distance as
\[ F_{\mathbf{M}}(\rho,\sigma)=\inf_{\mathcal{M}\in\mathbf{M}}F\left(\mathcal{M}(\rho),\mathcal{M}(\sigma)\right). \]
We have the well-known relations between these two distances (see e.g.~\cite{NC}, Chapter 9)
\begin{equation} \label{eq:trace-fidelity}
1-F_{\mathbf{M}}\leq D_{\mathbf{M}} \leq \left(1-F_{\mathbf{M}}^2\right)^{1/2}.
\end{equation}
We further define, for any set of states $\mathcal{K}$ on $\cH$, the measured by $\mathbf{M}$ trace-norm distance of $\rho$ to $\mathcal{K}$ as
\[ D_{\mathbf{M}}\left(\rho,\mathcal{K}\right)=\inf_{\sigma\in\mathcal{K}}D_{\mathbf{M}}(\rho,\sigma), \]
and the measured by $\mathbf{M}$ fidelity distance of $\rho$ to $\mathcal{K}$ as
\[ F_{\mathbf{M}}\left(\rho,\mathcal{K}\right)=\sup_{\sigma\in\mathcal{K}}F_{\mathbf{M}}(\rho,\sigma). \]

In the sequel, we shall consider the case where $\cH=\mathrm{A}\otimes \mathrm{B}$ is a tensor product Hilbert space, with $|\mathrm{A}|,|\mathrm{B}|<+\infty$. In this setting, we denote by $\mathcal{S}$ the set of separable states and by $\mathbf{SEP}$ the set of separable POVMs on $\cH$ (in the bipartite cut $\mathrm{A}{:}\mathrm{B}$).

\begin{lemma} \label{lemma:F_SEP}
Let $\mathrm{A}_1,\mathrm{B}_1,\mathrm{A}_2,\mathrm{B}_2$ be Hilbert spaces, and let $\rho_1$ be a state on $\mathrm{A}_1\otimes \mathrm{B}_1$, $\rho_2$ be a state on $\mathrm{A}_2\otimes \mathrm{B}_2$. Then,
\[ F\big(\rho_1\otimes\rho_2,\mathcal{S}(\mathrm{A}_1\mathrm{A}_2{:}\mathrm{B}_1\mathrm{B}_2)\big) \leq F_{\mathbf{SEP}}\big(\rho_1,\mathcal{S}(\mathrm{A}_1{:}\mathrm{B}_1)\big) F\big(\rho_2,\mathcal{S}(\mathrm{A}_2{:}\mathrm{B}_2)\big). \]
\end{lemma}

\begin{proof}
The proof is directly inspired from \cite{Piani}, adapted here to the case of fidelities rather than relative entropies.

Let $\mathcal{M}_1\equiv \big(M_1^{(i)}\big)_{i\in I}\in\mathbf{SEP}(\mathrm{A}_1{:}\mathrm{B}_1)$. Then, by monotonicity of the fidelity under the CPTP map $\mathcal{M}_1\otimes\mathcal{I}_2$, we have
\begin{align*} \sup_{\sigma_{12}\in\mathcal{S}(\mathrm{A}_1\mathrm{A}_2{:}\mathrm{B}_1\mathrm{B}_2)} F\left(\rho_1\otimes\rho_2,\sigma_{12}\right) \leq & \sup_{\sigma_{12}\in\mathcal{S}(\mathrm{A}_1\mathrm{A}_2{:}\mathrm{B}_1\mathrm{B}_2)} F\left(\mathcal{M}_1\otimes\mathcal{I}_2(\rho_1\otimes\rho_2),\mathcal{M}_1\otimes\mathcal{I}_2(\sigma_{12})\right)\\
= & F\left(\mathcal{M}_1\otimes\mathcal{I}_2(\rho_1\otimes\rho_2),\mathcal{M}_1\otimes\mathcal{I}_2(\widetilde{\sigma}_{12})\right), \end{align*}
for some $\widetilde{\sigma}_{12}\in\mathcal{S}(\mathrm{A}_1\mathrm{A}_2{:}\mathrm{B}_1\mathrm{B}_2)$. And,
\[ F\left(\mathcal{M}_1\otimes\mathcal{I}_2(\rho_1\otimes\rho_2), \mathcal{M}_1\otimes\mathcal{I}_2(\widetilde{\sigma}_{12})\right) = \sum_{i\in I}\sqrt{\Tr\left(M_1^{(i)}\rho_1\right)}\sqrt{\Tr\left(M_1^{(i)}\widetilde{\sigma}_1\right)} F\left(\rho_2,\widetilde{\sigma}_2^{(i)}\right), \]
where $\widetilde{\sigma}_1=\Tr_{\mathrm{A}_2\mathrm{B}_2}\left(\widetilde{\sigma}_{12}\right)\in\mathcal{S}(\mathrm{A}_1{:}\mathrm{B}_1)$ and for all $i\in I$, $\widetilde{\sigma}_2^{(i)}= \Tr_{\mathrm{A}_1\mathrm{B}_1}\big(M_1^{(i)}\otimes\Id_2 \widetilde{\sigma}_{12}\big)/\Tr_{\mathrm{A}_1\mathrm{B}_1}\big(M_1^{(i)}\widetilde{\sigma}_1\big)\in\mathcal{S}(\mathrm{A}_2{:}\mathrm{B}_2)$. Hence, for all $i\in I$, $F\big(\rho_2,\widetilde{\sigma}_2^{(i)}\big)\leq \sup_{\sigma_2\in\mathcal{S}(\mathrm{A}_2{:}\mathrm{B}_2)} F\left(\rho_2,\sigma_2\right)$, and subsequently
\begin{align*}
\sum_{i\in I}\sqrt{\Tr\left(M_1^{(i)}\rho_1\right)}\sqrt{\Tr\left(M_1^{(i)}\widetilde{\sigma}_1\right)} F\left(\rho_2,\widetilde{\sigma}_2^{(i)}\right) \leq & \left(\sup_{\sigma_2\in\mathcal{S}_{\mathrm{A}_2:\mathrm{B}_2}} F\left(\rho_2,\sigma_2\right) \right) F\left(\mathcal{M}_1(\rho_1),\mathcal{M}_1(\widetilde{\sigma}_1)\right)\\
\leq & \left(\sup_{\sigma_2\in\mathcal{S}(\mathrm{A}_2{:}\mathrm{B}_2)} F\left(\rho_2,\sigma_2\right)\right) \left(\sup_{\sigma_1\in\mathcal{S}(\mathrm{A}_1{:}\mathrm{B}_1)} F\left(\mathcal{M}_1(\rho_1),\mathcal{M}_1(\sigma_1)\right)\right).
\end{align*}
We thus have shown that, for any $\mathcal{M}_1\in\mathbf{SEP}(\mathrm{A}_1{:}\mathrm{B}_1)$,
\[ F\left(\rho_1\otimes\rho_2,\mathcal{S}(\mathrm{A}_1\mathrm{A}_2{:}\mathrm{B}_1\mathrm{B}_2)\right) \leq \left(\sup_{\sigma_1\in\mathcal{S}(\mathrm{A}_1{:}\mathrm{B}_1)} F\left(\mathcal{M}_1(\rho_1),\mathcal{M}_1(\sigma_1)\right) \right) F\left(\rho_2,\mathcal{S}(\mathrm{A}_2{:}\mathrm{B}_2)\right). \]
Taking the infimum over $\mathcal{M}_1\in\mathbf{SEP}(\mathrm{A}_1{:}\mathrm{B}_1)$, we get precisely the statement in Lemma \ref{lemma:F_SEP}.
\end{proof}

\begin{theorem}
  \label{th:F_SEP}
  Let $\mathrm{A},\mathrm{B}$ be Hilbert spaces, and let $\rho$ be a state on $\mathrm{A}\otimes\mathrm{B}$.
  Then, for any $n\in\N$,
  \[
    F\big(\rho^{\otimes n},\mathcal{S}(\mathrm{A}^n{:}\mathrm{B}^n)\big)
               \leq F_{\mathbf{SEP}}\big(\rho,\mathcal{S}(\mathrm{A}{:}\mathrm{B})\big)^n.
  \]
\end{theorem}

\begin{proof}
Theorem \ref{th:F_SEP} is a direct corollary of Lemma \ref{lemma:F_SEP}, obtained by iterating the latter.
\end{proof}

\subsection{Weak multiplicativity of $h_{sep}$}
With these facts prepared, we can now derive our main theorem.

\begin{theorem}
  \label{th:h_sep-n}
  Let $M$ be an operator on the tensor product Hilbert space $\mathrm{A}\otimes \mathrm{B}$, satisfying $0 \leq M \leq \Id$, and set $r=\|M\|_2$. If $h_{sep}(M)\leq 1-\delta$, for some $0<\delta<1$,
  then for any $n\in\N$,
  \[
    h_{sep}(M^{\otimes n}) \leq \left(1-\frac{\delta^2}{5r^2}\right)^n.
  \]
% Then, $h\left(P^{\otimes n},\mathcal{S}_{(\C^d)^{\otimes n}:(\C^d)^{\otimes n}}\right)\leq\left(1-\delta^2/33d^2\right)^n$.
\end{theorem}

\begin{proof}
Let $\rho\in\mathcal{S}(\mathrm{A}^{n}{:}\mathrm{B}^{n})$. Our goal will be first of all to show that $\Tr\left(M^{\otimes n}\rho\right)\leq 2(n+1)^{3|\mathrm{A}|^2|\mathrm{B}|^2}\left(1-\delta^2/2r^2\right)^n$.
%$\Tr\left(P^{\otimes n}\rho\right)\leq 2(n+1)^{3d^4}\left(1-\delta^2/33d^2\right)^n$.
Now, observe that
\[ \Tr\left(M^{\otimes n}\rho\right) = \Tr\left( \left(\frac{1}{n!}\sum_{\pi\in\mathcal{S}_n}U_{\pi}M^{\otimes n}U_{\pi}^{\dagger}\right) \rho\right) = \Tr\left(M^{\otimes n} \left(\frac{1}{n!}\sum_{\pi\in\mathcal{S}_n}U_{\pi}^{\dagger}\rho U_{\pi}\right) \right), \]
the first equality being by $n$-symmetry of $M^{\otimes n}$ and the second one by cyclicity
of the trace. Hence, for our purposes, we may actually assume without loss of generality
that $\rho\in\mathcal{S}(\mathrm{A}^{n}{:}\mathrm{B}^{n})$ is $n$-symmetric.

Yet, if $\rho$ is an $n$-symmetric state on $\left(\mathrm{A}\otimes\mathrm{B}\right)^{\otimes n}$,
we know by Theorem \ref{th:ps-mixed} that there exists a probability measure
$\mu$ on the set of states on $\mathrm{A}\otimes\mathrm{B}$ such that
\[
  \rho \leq (n+1)^{3|\mathrm{A}|^2|\mathrm{B}|^2} \int_{\sigma\in\mathcal{D}(\mathrm{A}\otimes\mathrm{B})}
                     F\left(\rho,\sigma^{\otimes n}\right)^2 \sigma^{\otimes n}\,\mathrm{d}\mu(\sigma).
\]
So, by multiplicativity of the trace on tensor products, we get in that case
\[
  \Tr\left(M^{\otimes n}\rho\right)
      \leq (n+1)^{3|\mathrm{A}|^2|\mathrm{B}|^2} \int_{\sigma\in\mathcal{D}(\mathrm{A}\otimes\mathrm{B})}
                     F\left(\rho,\sigma^{\otimes n}\right)^2 \Tr\left(M\sigma\right)^n\,\mathrm{d}\mu(\sigma).
\]
Consequently, for any $0<\epsilon<1$, setting $\mathcal{K}_{\epsilon}=\left\{\sigma\in\mathcal{D}(\mathrm{A}\otimes\mathrm{B}) \st \left\|\sigma-\mathcal{S}(\mathrm{A}{:}\mathrm{B})\right\|_2\leq \epsilon/r\right\}$, we have, upper bounding either $F\left(\rho,\sigma^{\otimes n}\right)$ or $\Tr\left(M\sigma\right)$ by $1$,
\[ \Tr\left(M^{\otimes n}\rho\right) \leq\, (n+1)^{3|\mathrm{A}|^2|\mathrm{B}|^2} \left( \int_{\sigma\in\mathcal{K}_{\epsilon}} \Tr\left(M\sigma\right)^n \mathrm{d}\mu(\sigma) + \int_{\sigma\notin\mathcal{K}_{\epsilon}} F\left(\rho,\sigma^{\otimes n}\right)^2 \mathrm{d}\mu(\sigma)\right). \]
Now, if $\sigma\in\mathcal{K}_{\epsilon}$, this means that there exists $\tau\in\mathcal{S}(\mathrm{A}{:}\mathrm{B})$ such that $\|\sigma-\tau\|_2\leq\epsilon/r$, so that
\[
  \Tr(M\sigma) =    \Tr(M\tau)+\Tr(M(\sigma-\tau))
               \leq \Tr(M\tau)+ \|M\|_2\|\sigma-\tau\|_2
               \leq 1-\delta+\epsilon.
\]
The next to last inequality is simply by the Cauchy--Schwarz inequality,
while the last one is by assumption on $M$, $\tau$, $\sigma$. And if
$\sigma\notin\mathcal{K}_{\epsilon}$, then
\[ F\left(\rho,\sigma^{\otimes n}\right)  \leq F\left(\sigma^{\otimes n}, \mathcal{S}(\mathrm{A}^{n}{:}\mathrm{B}^{n})\right) \leq F_{\mathbf{SEP}}\left(\sigma, \mathcal{S}(\mathrm{A}{:}\mathrm{B})\right)^n \leq \left(1-\frac{\epsilon^2}{4r^2}\right)^{n/2}. \]
The first inequality is because $\rho\in\mathcal{S}(\mathrm{A}^{n}{:}\mathrm{B}^{n})$, the second one is by Theorem \ref{th:F_SEP}, and the third one is obtained by combining equation \eqref{eq:trace-fidelity} with the known lower bound $\left\|\sigma-\mathcal{S}(\mathrm{A}{:}\mathrm{B})\right\|_{\mathbf{SEP}}\geq \left\|\sigma-\mathcal{S}(\mathrm{A}{:}\mathrm{B})\right\|_2$ (see e.g.~\cite{LaW1}).

Putting everything together, we obtain in the end that for any $0<\epsilon<1$,
\begin{equation} \label{eq:epsilon} \Tr\left(M^{\otimes n}\rho\right) \leq\, (n+1)^{3|\mathrm{A}|^2|\mathrm{B}|^2} \left( (1-\delta+\epsilon)^n + \left(1-\frac{\epsilon^2}{4r^2}\right)^n \right). \end{equation}
In particular, choosing $\epsilon=2r^2\left((1+\delta/r^2)^{1/2}-1\right)$ in equation \eqref{eq:epsilon}, so that $\epsilon^2/4r^2=\delta-\epsilon\geq \delta^2/5r^2$, we get $\Tr\left(M^{\otimes n}\rho\right) \leq 2(n+1)^{3|\mathrm{A}|^2|\mathrm{B}|^2} \left(1-\delta^2/5r^2\right)^n$. And consequently
\begin{equation}
  \label{eq:h_SEPn}
  h_{sep}(M^{\otimes n}) \leq 2(n+1)^{3|\mathrm{A}|^2|\mathrm{B}|^2}\left(1-\frac{\delta^2}{5r^2}\right)^n.
\end{equation}

In order to conclude, we just need to remove the polynomial pre-factor in equation \eqref{eq:h_SEPn}.
Assume that there exists a constant $C>0$ such that
$h_{sep}(M^{\otimes N})\geq C\left(1-\delta^2/5r^2\right)^N$ for some $N\in\N$.
Then, we would have for any $n\in\N$,
\[
  h_{sep}\left(M^{\otimes Nn}\right) \geq h_{sep}(M^{\otimes N})^n
                                     \geq C^n\left(1-\frac{\delta^2}{5r^2}\right)^{Nn}.
\]
On the other hand, equation \eqref{eq:h_SEPn} says that we also have
\[
  h_{sep}\left(M^{\otimes Nn}\right) \leq 2(Nn+1)^{3|\mathrm{A}|^2|\mathrm{B}|^2} \left(1-\frac{\delta^2}{5r^2}\right)^{Nn}.
\]
Letting $n$ grow, we see that the only option to make these two inequalities
compatible is to have $C\leq 1$, which is precisely what we wanted to show.
\end{proof}

The conclusion of Theorem \ref{th:h_sep-n} had already been obtained via completely
different techniques than the one presented here (and even with slightly better constants).
However, the good thing about the de Finetti reduction approach is that it
gives, almost for free, not only this exponential decay result for the behaviour
of $h_{sep}$ under tensoring, but also some kind of concentration statement.
To be precise, assume that $M$ is an operator on $\A\otimes\B$, satisfying $0 \leq M \leq \Id$
and $h_{sep}(M)\leq 1-\delta$ for some $0<\delta<1$.
Then, $M$ can be identified with a binary test that a separable state is
guaranteed to pass only with probability $h_{sep}(M)\leq 1-\delta$,
while there exists some (entangled) state that would pass it with probability $h_{all}(M)=\|M\|_{\infty}$,
which may be $1$.
Hence, a natural question would be: performing this test $n$ times in parallel,
what is the probability that a separable state passes a certain fraction $t/n$ of them?
Such maximum probability is nothing else than $h_{sep}\left(M^{(t/n)}\right)$, where the operator $M^{(t/n)}$ on $(\A\otimes\B)^{\otimes n}$ is defined as
\[ M^{(t/n)} = \sum_{I\subset[n],\,|I|\geq t}M^{\otimes I}\otimes\Id^{\otimes I^c} . \]
Obviously, if $t < (1-\delta)n$ then the answer is asymptotically $1$,
whereas for $t=n$ the answer is $h_{sep}(M^{\otimes n})$, which decays
exponentially fast with $n$ as established in Theorem \ref{th:h_sep-n}.
But is such exponential amplification of the failing probability already
true for $t$ just slightly above $(1-\delta)n$? Theorem \ref{th:h_sep-t/n}
answers this question affirmatively.

\begin{theorem}
  \label{th:h_sep-t/n}
  Let $M$ be an operator on the tensor product Hilbert space $\mathrm{A}\otimes \mathrm{B}$,
  satisfying $0 \leq M \leq \Id$, and set $r=\|M\|_2$.
  If $h_{sep}(M)\leq 1-\delta$ for some $0<\delta<1$, then for any
  $n,t\in\N$ with $t\geq(1-\delta+\alpha)n$ for some $0<\alpha\leq\delta$, we have
  \[
    h_{sep}\left(M^{(t/n)}\right)\leq\exp\left(-n\frac{\alpha^2}{5r^2}\right).
  \]
% Then, $h\left(P^{\otimes n},\mathcal{S}_{(\C^d)^{\otimes n}:(\C^d)^{\otimes n}}\right)\leq\left(1-\delta^2/33d^2\right)^n$.
\end{theorem}

\begin{proof}
Following the exact same lines as in the proof of Theorem \ref{th:h_sep-n}, we now have in place of equation \eqref{eq:epsilon}
\begin{equation} \label{eq:epsilon'} h_{sep}\left(M^{(t/n)}\right) \leq\, (n+1)^{3|\mathrm{A}|^2|\mathrm{B}|^2} \left( \exp\left[-2n(\alpha-\epsilon)^2\right] + \exp\left[-n\frac{\epsilon^2}{4r^2}\right] \right). \end{equation}
This is indeed a consequence of Hoeffding's inequality (and of the fact that $e^{-x}\geq 1-x$ for any $x>0$). So in particular, choosing $\epsilon=\alpha\left(1-(\sqrt{2}-1)/(8r^2-1)\right)$ in equation \eqref{eq:epsilon'}, so that $\epsilon^2/4r^2=2(\alpha-\epsilon)^2\geq\alpha^2/5r^2$, and removing the polynomial pre-factor by the same trick as in the proof of Theorem \ref{th:h_sep-n}, we get as announced
\[ h_{sep}\left(M^{(t/n)}\right)\leq\exp\left(-n\frac{\alpha^2}{5r^2}\right). \qedhere \]
\end{proof}

\section{Exponential decay and concentration of $h_{sep}$ via entanglement measure approach}
\label{sec:sep2}

\subsection{Quantifying the disturbance induced by measurements}

Let us state first a few technical lemmas that we will need later on to establish our main result. In what follows, we will use a few standard definitions from quantum Shannon theory, which we recall here: The entropy of a state $\rho$ is defined as $S(\rho)=-\Tr(\rho\log\rho)$. From there, one can define the mutual information of a bipartite state $\rho_{\A\B}$ and the conditional mutual information of a tripartite state $\rho_{\A\B\mathrm{C}}$ as, respectively,
\begin{align*}
& I(\A:\B)_{\rho}=S(\rho_{\A})+S(\rho_{\B})-S(\rho_{\A\B}),\\
& I(\A:\B|\mathrm{C})_{\rho}= S(\rho_{\A\mathrm{C}})+S(\rho_{\B\mathrm{C}})-S(\rho_{\mathrm{C}})-S(\rho_{\A\B\mathrm{C}}).
\end{align*}
%Both of them are non-negative due to (strong) sub-additivity of the entropy.
Finally, the relative entropy between states $\rho$ and $\sigma$ is defined as $D(\rho\|\sigma)=\Tr(\rho(\log\rho-\log\sigma))$.

\begin{lemma}
\label{lemma:post-POVM}
Let $\rho$ be a state on $\mathrm{U}\otimes\mathrm{V}$ and let $T$ be an operator on $\mathrm{U}$, satisfying $0\leq T\leq\Id$. Next, define $p=\Tr_{\mathrm{U}\mathrm{V}}\left[\left(T_{\mathrm{U}}\otimes\Id_{\mathrm{V}}\right)\rho_{\mathrm{U}\mathrm{V}}\right]$ as the probability of obtaining the first outcome when the two-outcome POVM $\left(T_{\mathrm{U}}\otimes\Id_{\mathrm{V}}, (\Id_{\mathrm{U}}-T_{\mathrm{U}})\otimes\Id_{\mathrm{V}}\right)$ is performed on $\rho_{\mathrm{U}\mathrm{V}}$, and $\tau_{\mathrm{V}}=\Tr_{\mathrm{U}}\left[\left(T_{\mathrm{U}}\otimes\Id_{\mathrm{V}}\right)\rho_{\mathrm{U}\mathrm{V}}\right]/p$ as the corresponding post-measurement state on $\mathrm{V}$. Also, denote by $\rho_{\mathrm{V}}=\Tr_{\mathrm{U}}\left[\rho_{\mathrm{U}\mathrm{V}}\right]$ the reduced state of $\rho_{\mathrm{U}\mathrm{V}}$ on $\mathrm{V}$. Then,
\[ D\left(\tau_{\mathrm{V}}\big\|\rho_{\mathrm{V}}\right) \leq - \log p.\]
\end{lemma}

\begin{proof}
Note that $\rho_{\mathrm{V}}= p\tau_{\mathrm{V}} +(1-p)\sigma_{\mathrm{V}}$, where $\sigma_{\mathrm{V}}=\Tr_{\mathrm{U}}\left[\left((\Id_{\mathrm{U}}-T_{\mathrm{U}})\otimes\Id_{\mathrm{V}}\right) \rho_{\mathrm{U}\mathrm{V}}\right]/(1-p)$. We therefore have the operator inequality $p\tau_{\mathrm{V}}\leq\rho_{\mathrm{V}}$. And hence,
\[ D\left(\tau_{\mathrm{V}}\big\|\rho_{\mathrm{V}}\right) = \Tr\left[\tau_{\mathrm{V}}\left(\log\tau_{\mathrm{V}}-\log\rho_{\mathrm{V}}\right)\right] \leq \Tr\left[\tau_{\mathrm{V}}\left(\log\tau_{\mathrm{V}}-\log (p\tau_{\mathrm{V}})\right)\right] = -\log p, \]
the next to last inequality being because $\log$ is an operator monotone function.
\end{proof}

Let us recall the definition of the squashed entanglement $E_{sq}$, introduced in \cite{CW}:
\[ E_{sq}\left(\rho_{\mathrm{A}\mathrm{B}}\right) = \inf\left\{ \frac{1}{2} I(\mathrm{A}:\mathrm{B}|\mathrm{E})_{\rho} \st \Tr_{\mathrm{E}}\left(\rho_{\mathrm{A}\mathrm{B}\mathrm{E}}\right)=\rho_{\mathrm{A}\mathrm{B}} \right\}. \]

\begin{lemma} \label{lemma:E_sq}
Let $M_{\mathrm{A}\mathrm{B}}$ be an operator on the tensor product Hilbert space $\mathrm{A}\otimes\mathrm{B}$, satisfying $0\leq M_{\mathrm{A}\mathrm{B}}\leq \Id$, and let $\alpha_{\mathrm{A}^n},\beta_{\mathrm{B}^n}$ be states on $\mathrm{A}^{\otimes n},\mathrm{B}^{\otimes n}$ respectively. Next, fix $1\leq k\leq n-1$, and define
\[ p_k=\Tr_{\mathrm{A}^n\mathrm{B}^n}\left[\left(M_{\mathrm{A}\mathrm{B}}^{\otimes k}\otimes\Id_{\mathrm{A}\mathrm{B}}^{\otimes n-k}\right)\alpha_{\mathrm{A}^n}\otimes\beta_{\mathrm{B}^n}\right]\ \text{and}\ \tau(k)_{\mathrm{A}^{n-k}\mathrm{B}^{n-k}}=\frac{1}{p_k}\Tr_{\mathrm{A}^k\mathrm{B}^k}\left[\left(M_{\mathrm{A}\mathrm{B}}^{\otimes k}\otimes\Id_{\mathrm{A}\mathrm{B}}^{\otimes n-k}\right)\alpha_{\mathrm{A}^n}\otimes\beta_{\mathrm{B}^n}\right].\]
Then,
\[ \sum_{j=k+1}^n E_{sq}\left(\tau(k)_{\mathrm{A}_j\mathrm{B}_j}\right) \leq \frac{1}{2}\log \frac{1}{p_k} .\]
\end{lemma}

\begin{proof}
By Lemma \ref{lemma:post-POVM}, with $\mathrm{U}=\mathrm{A}^{\otimes k}\otimes\mathrm{B}^{\otimes k}$, $\mathrm{V}=\mathrm{A}^{\otimes n-k}\otimes\mathrm{B}^{\otimes n-k}$, $T_{\mathrm{U}}=M_{\mathrm{A}\mathrm{B}}^{\otimes k}$ and $\rho_{\mathrm{U}\mathrm{V}}=\alpha_{\mathrm{A}^n}\otimes\beta_{\mathrm{B}^n}$, we have
\[ D\left(\tau(k)_{\mathrm{A}^{n-k}\mathrm{B}^{n-k}}\big\|\alpha_{\mathrm{A}^{n-k}}\otimes\beta_{\mathrm{B}^{n-k}}\right) \leq \log \frac{1}{p_k}. \]
Now, observe that
\begin{align*}
D\left(\tau(k)_{\mathrm{A}^{n-k}\mathrm{B}^{n-k}}\big\|\alpha_{\mathrm{A}^{n-k}}\otimes\beta_{\mathrm{B}^{n-k}}\right) \geq & \, D\left(\tau(k)_{\mathrm{A}^{n-k}\mathrm{B}^{n-k}}\big\|\tau(k)_{\mathrm{A}^{n-k}}\otimes\tau(k)_{\mathrm{B}^{n-k}}\right) \\
= & \, I\left(\mathrm{A}_{k+1}\ldots \mathrm{A}_n:\mathrm{B}_{k+1}\ldots \mathrm{B}_n\right)_{\tau(k)}\\
= & \sum_{j=k+1}^n I\left(\mathrm{A}_j:\mathrm{B}_{k+1}\ldots \mathrm{B}_n|\mathrm{A}_{k+1}\ldots \mathrm{A}_{j-1}\right)_{\tau(k)}\\
\geq & \sum_{j=k+1}^n I\left(\mathrm{A}_j:\mathrm{B}_j|\mathrm{A}_{k+1}\ldots \mathrm{A}_{j-1}\right)_{\tau(k)} \\
\geq & \sum_{j=k+1}^n 2\,E_{sq}\left(\tau(k)_{\mathrm{A}_j\mathrm{B}_j}\right).
\end{align*}
The first inequality is due to the fact that, given a bipartite state $\tau_{\mathrm{U}\mathrm{V}}$ on $\mathrm{U}\otimes\mathrm{V}$, for any states $\rho_{\mathrm{U}}$, $\rho_{\mathrm{V}}$ on $\mathrm{U}$, $\mathrm{V}$ respectively, $D(\tau_{\mathrm{U}\mathrm{V}}\|\rho_{\mathrm{U}}\otimes\rho_{\mathrm{V}})\geq D(\tau_{\mathrm{U}\mathrm{V}}\|\tau_{\mathrm{U}}\otimes\tau_{\mathrm{V}})$.
%The second equality is by definition of the relative entropy and the quantum mutual information.
The third equality and the fourth inequality follow from the chain rule and the monotonicity under discarding of subsystems, respectively, for the quantum mutual information.
And the last inequality is by definition of the squashed entanglement.
\end{proof}

\begin{remark}
\label{remark:E_I}
Observe that under the assumptions of Lemma \ref{lemma:E_sq},
we actually have the stronger conclusion
\[
  \sum_{j=k+1}^n E_{I}\left(\tau(k)_{\mathrm{A}_j\mathrm{B}_j}\right) \leq \frac{1}{2}\log \frac{1}{p_k},
\]
where $E_I$ is the \emph{conditional entanglement of mutual information (CEMI)}
introduced in \cite{HWY}:
\[
  E_I\left(\rho_{\mathrm{A}\mathrm{B}}\right) = \inf \left\{ \frac{1}{2}\big[ I(\mathrm{A}\mathrm{A}':\mathrm{B}\mathrm{B}')_{\rho}
                                    - I(\mathrm{A}':\mathrm{B}')_{\rho}\big]
                              \st \Tr_{\mathrm{A}'\mathrm{B}'}\left(\rho_{\mathrm{A}\mathrm{A}'\mathrm{B}\mathrm{B}'}\right) =\rho_{\mathrm{A}\mathrm{B}} \right\}.
\]
CEMI is always at least as large as squashed entanglement: for any state $\rho_{\mathrm{A}\mathrm{B}}$, $E_I(\rho_{\mathrm{A}\mathrm{B}}) \geq E_{sq}(\rho_{\mathrm{A}\mathrm{B}})$. But the precise relation between these two entanglement measures is unknown.
In~\cite{HWY} it was furthermore shown that, like squashed entanglement, CEMI is additive,
and more generally super-additive in the sense that
\[
  E_I(\rho_{\mathrm{A}_1\mathrm{A}_2{:}\mathrm{B}_1\mathrm{B}_2}) \geq E_I(\rho_{\mathrm{A}_1{:}\mathrm{B}_1}) + E_I(\rho_{\mathrm{A}_2{:}\mathrm{B}_2}).
\]
However, unlike squashed entanglement, there is no simple proof of
monogamy of CEMI, and it may well not hold in general.
\end{remark}

\begin{lemma}[Cf.~\cite{Holenstein}, Lemma 8.6]
\label{lemma:recursivity}
Let $0<\nu<1$ and $c>0$. Let also $n\in\N$ and assume that $(p_k)_{1\leq k\leq n}$ is a sequence of numbers satisfying $1> p_1\geq\cdots\geq p_n >0$ and
\[ \forall\ 1\leq k\leq n-1,\ p_{k+1} \leq p_k \left( \sqrt{\frac{c}{n-k}\log\frac{1}{p_k}} + \nu \right). \]
Then, for any $0<\gamma<1-\nu$ such that $p_1\leq \nu+\gamma$, we have
\[ \forall\ 1\leq k\leq n,\ p_k\leq (\nu+\gamma)^{\min(k,k_0)},\ \text{where}\ k_0=\frac{\gamma^2}{c\log[1/(\nu+\gamma)]+\gamma^2}(n+1). \]
\end{lemma}

\begin{proof}
To prove Lemma \ref{lemma:recursivity}, we only have to show that
\begin{equation} \label{eq:rec} \forall\ 1\leq k\leq k_0,\ p_k\leq (\nu+\gamma)^k. \end{equation}
Indeed, the case $k>k_0$ then directly follows from the assumption that the sequence $(p_k)_{1\leq k\leq n}$ is non-increasing, so that $p_k\leq p_{k_0} \leq (\nu+\gamma)^{k_0}$.

Let us establish \eqref{eq:rec} by recursivity. The statement obviously holds for $k=1$ since $p_1\leq \nu+\gamma$ by hypothesis. So assume next that it holds for some $k\leq k_0-1$. If $p_k\leq(\nu+\gamma)^{k+1}$, then clearly $p_{k+1}\leq p_k\leq(\nu+\gamma)^{k+1}$. Otherwise, by the way $p_{k+1}$ is related to $p_k$, we then have $p_{k+1}\leq (\nu+\gamma)^k\big(\sqrt{c(k+1)\log[1/(\nu+\gamma)]/(n-k)}+\nu\big)$, and the latter quantity is smaller than $(\nu+\gamma)^{k+1}$ if $(k+1)/(n-k)\leq \gamma^2/\left(c\log[1/(\nu+\gamma)]\right)$, which can be checked to be equivalent to $k+1\leq k_0$. Hence in both cases, the statement holds for $k+1$.
\end{proof}

\begin{corollary}[Cf. \cite{Holenstein}, Lemma 8.6]
\label{cor:recursivity}
Let $(p_k)_{1\leq k\leq n}$ be a sequence of numbers satisfying the assumptions of Lemma \ref{lemma:recursivity}, and the additional condition $p_1\leq 1-(1-\nu)/2$. Then,
\[ p_n\leq \left(1-\frac{(1-\nu)^2}{8 c}\right)^n. \]
\end{corollary}

\begin{proof}
Corollary \ref{cor:recursivity} follows from applying Lemma \ref{lemma:recursivity} in the particular case $\gamma=(1-\nu)/2$. Indeed, we then have $\nu+\gamma=1-(1-\nu)/2=(1+\nu)/2$, so that
\[ k_0= \frac{(1-\nu)^2}{4c\log[2/(1+\nu)]+(1-\nu)^2}(n+1) \geq \frac{(1-\nu)^2}{4c\log[2^{1-\nu}]+1}(n+1) \geq \frac{1-\nu}{4c}n, \]
And consequently,
\[ p_n \leq \left(1-\frac{1-\nu}{2}\right)^{n(1-\nu)/(4c)} \leq \left(1-\frac{(1-\nu)^2}{8 c}\right)^n, \]
which is exactly the announced upper bound for $p_n$.
\end{proof}

\subsection{Weak multiplicativity of $h_{sep}$}

Our approach in this section, to prove the multiplicative behaviour of $h_{sep}$, is directly inspired from the seminal angle of attack to the parallel repetition problem for classical non-local games: our Theorem \ref{th:h_sep-h_qext} is an analogue of the exponential decay results by Raz \cite{Raz} and Holenstein \cite{Holenstein}, while our Theorem \ref{th:w_sep-h_qext} is an analogue of the concentration bound result by Rao \cite{Rao}. Indeed, here in the same spirit as theirs, we want to make precise the following intuition: if the initial state of a system on $(\A\otimes\B)^{\otimes n}$ is product across the cut $\A^n{:}\B^n$, then performing a measurement $(M,\Id-M)$ on a few subsystems $\A\otimes\B$ only should not create too much correlations in the post-measurement state on the remaining subsystems.

Before we prove the main result of this section,
we need to recall one last definition: For any $q\in\N$,
a state $\rho_{\mathrm{A}\mathrm{B}}$ on a bipartite Hilbert space $\mathrm{A}\otimes \mathrm{B}$
is said to be $q$-extendible with respect to $\mathrm{B}$
if there exists a state $\rho_{\mathrm{A}\mathrm{B}^q}$ on $\mathrm{A}\otimes \mathrm{B}^{\otimes q}$
that is invariant under any permutation of the $\mathrm{B}$-subsystems and such
that $\rho_{\mathrm{A}\mathrm{B}}=\Tr_{\mathrm{B}^{q-1}}\rho_{\mathrm{A}\mathrm{B}^q}$.
We shall denote by $\mathcal{E}_q(\mathrm{A}{:}\mathrm{B})$ the set of
$q$-extendible states with respect to $\mathrm{B}$ on $\mathrm{A}\otimes \mathrm{B}$, and by $h_{q-ext}$ its associated support function.

\begin{theorem}
\label{th:h_sep-h_qext}
Let $M$ be an operator on the tensor product Hilbert space $\mathrm{A} \otimes \mathrm{B}$, satisfying $0 \leq M \leq \Id$. Then, for any $q\in\N$,
\begin{equation}
  \label{eq:qext-bound}
  h_{sep}\left(M^{\otimes n}\right) \leq \left( 1-\frac{\left(1-h_{q-ext}(M)\right)^2}{8\ln 2\,q^2}\right)^n . \end{equation}
And consequently, if $h_{sep}(M)\leq 1-\delta$ for some $0<\delta<1$, then
\begin{equation}
  \label{eq:sep-bound}
  h_{sep}\left(M^{\otimes n}\right) \leq \left(1 -\frac{\delta^4}{512\ln 2\,d^4}\right)^n,
\end{equation}
assuming $|\A|=|\B|=d$.
\end{theorem}

\begin{proof}
To establish the first statement \eqref{eq:qext-bound}, we have to show that,
\[
  \forall\ \rho_{\mathrm{A}^n\mathrm{B}^n}\in\mathcal{S}(\mathrm{A}^{n}{:}\mathrm{B}^{n}),\ \Tr\left(M_{\mathrm{A}\mathrm{B}}^{\otimes n}\rho_{\mathrm{A}^n\mathrm{B}^n} \right) \leq \left( 1-\frac{\left(1-h_{q-ext}(M_{\mathrm{A}\mathrm{B}})\right)^2}{8\ln 2\,q^2}\right)^n .
\]
Note that, with this aim in view, we can without loss of generality focuss only on states which are extremal in $\mathcal{S}(\mathrm{A}^{n}{:}\mathrm{B}^{n})$, namely on states which are product across the cut $\mathrm{A}^{\otimes n}{:}\mathrm{B}^{\otimes n}$. So let $\alpha_{\mathrm{A}^n}\otimes\beta_{\mathrm{B}^n}$ be such a state, and set $p_0=1$, $\tau(0)_{\mathrm{A}^n\mathrm{B}^n}=\alpha_{\mathrm{A}^n}\otimes\beta_{\mathrm{B}^n}$. In the sequel, we will use the following notation: given $I_k\subset[n]$ with $|I_k|=k$, define $M^{(I_k)}_{\mathrm{A}^n\mathrm{B}^n}$ as
\[ M^{(I_k)}_{\mathrm{A}^n\mathrm{B}^n} = M_{\mathrm{A}\mathrm{B}}^{\otimes I_k}\otimes\Id_{\mathrm{A}\mathrm{B}}^{\otimes I_k^c}. \]
Then, for each $1\leq k\leq n$, construct recursively
\begin{align*}
& p_{k}= \Tr_{\mathrm{A}^n\mathrm{B}^n} \left[M^{(I_k)}_{\mathrm{A}^n\mathrm{B}^n}\alpha_{\mathrm{A}^n}\otimes\beta_{\mathrm{B}^n}\right], \\
& \tau(k)_{\mathrm{A}_{I_k^c}\mathrm{B}_{I_k^c}}= \frac{1}{p_k} \Tr_{\mathrm{A}_{I_k}\mathrm{B}_{I_k}} \left[M^{(I_k)}_{\mathrm{A}^n\mathrm{B}^n}\alpha_{A^n}\otimes\beta_{B^n}\right],
\end{align*}
with $i_{k}$ chosen in $I_{k-1}^c$ such that
\[ E_{sq}\left(\tau(k-1)_{\mathrm{A}_{i_k}\mathrm{B}_{i_k}}\right)\leq \frac{1}{n-k+1}\,\frac{1}{2}\log \frac{1}{p_{k-1}}. \]
We know that this is possible. Indeed, assuming that $p_{k-1}$, $\tau(k-1)_{\mathrm{A}_{I_{k-1}^c}\mathrm{B}_{I_{k-1}^c}}$ have been constructed, Lemma \ref{lemma:E_sq} guarantees that
\[ \frac{1}{n-k+1}\sum_{j=1}^{n-k+1}E_{sq}\left(\tau(k-1)_{\mathrm{A}_{I_{k-1}^c}\mathrm{B}_{I_{k-1}^c}}\right)\leq \frac{1}{n-k+1}\,\frac{1}{2}\log \frac{1}{p_{k-1}}, \]
so that there necessarily exists an index $i\in I_{k-1}^c$ such that $E_{sq}\big(\tau(k-1)_{\mathrm{A}_{i}\mathrm{B}_{i}}\big)$ is smaller than the quantity on the right-hand-side of the average upper bound above.

Now, notice that the $p_k$, $0\leq k\leq n$, are related by the recursion formula
\[ \forall\ 0\leq k\leq n-1,\ p_{k+1} = p_k\Tr_{\mathrm{A}_{i_{k+1}}\mathrm{B}_{i_{k+1}}}\left(M_{\mathrm{A}_{i_{k+1}}\mathrm{B}_{i_{k+1}}} \tau(k)_{\mathrm{A}_{i_{k+1}}\mathrm{B}_{i_{k+1}}}\right), \]
where, by the way the $\tau(k)$, $0\leq k\leq n$, are built
\[ E_{sq}\left(\tau(k)_{\mathrm{A}_{i_{k+1}}\mathrm{B}_{i_{k+1}}}\right) \leq \frac{1}{n-k}\frac{1}{2}\log \frac{1}{p_k}. \]
Yet, we know from \cite{LiW} that this implies that
\begin{equation} \label{eq:E_sq-q-ext} \exists\ \sigma_{\mathrm{A}\mathrm{B}}\in\mathcal{E}_q(\mathrm{A}{:}\mathrm{B}):\ \left\| \tau(k)_{\mathrm{A}\mathrm{B}} - \sigma_{\mathrm{A}\mathrm{B}} \right\|_1 \leq \sqrt{2\ln 2}(q-1)\sqrt{\frac{1}{n-k}\frac{1}{2}\log \frac{1}{p_k}} \leq \sqrt{\frac{\ln 2\, q^2}{n-k}\log \frac{1}{p_k}}. \end{equation}
And therefore,
\[ p_{k+1} \leq p_k \left(\left\|M_{\mathrm{A}\mathrm{B}}\right\|_{\infty}\left\| \tau(k)_{\mathrm{A}\mathrm{B}} - \sigma_{\mathrm{A}\mathrm{B}} \right\|_1 + \Tr\left(M_{\mathrm{A}\mathrm{B}}\sigma_{\mathrm{A}\mathrm{B}}\right) \right) \leq p_k \left(\sqrt{\frac{\ln 2\, q^2}{n-k}\log \frac{1}{p_k}} + h_{q-ext}\left(M_{\mathrm{A}\mathrm{B}}\right) \right). \]
With this upper bound, and because we also clearly have $1> p_1\geq\cdots\geq p_n> 0$ as well as the requirement $p_1\leq h_{sep}(M)\leq h_{q-ext}(M)\leq 1-(1-h_{q-ext})/2$, it follows from Corollary \ref{cor:recursivity} that
\[ \Tr\left(M_{\mathrm{A}\mathrm{B}}^{\otimes n}\alpha_{\A^n}\otimes\beta_{\B^n}\right)=p_n \leq \left( 1-\frac{\left(1-h_{q-ext}(M_{\mathrm{A}\mathrm{B}})\right)^2}{8\ln 2\,q^2}\right)^n, \]
which is precisely what we wanted to prove.

From there, the second statement \eqref{eq:sep-bound} easily follows.
Indeed, in the case where $|\mathrm{A}|=|\mathrm{B}|=d$, we know from \cite{CKMR} that,
for any $q\in\N$, $\rho_{\mathrm{A}\mathrm{B}}\in\mathcal{E}_q(\mathrm{A}{:}\mathrm{B})$ implies that
there exists $\sigma_{\mathrm{A}\mathrm{B}}\in\mathcal{S}(\mathrm{A}{:}\mathrm{B}):\ \|\rho_{\mathrm{A}\mathrm{B}}-\sigma_{\mathrm{A}\mathrm{B}}\|_1\leq 2d^2/q$,
so that $h_{q-ext}\leq h_{sep}+2d^2/q$. Hence, if $h_{sep}(M_{\mathrm{A}\mathrm{B}})\leq 1-\delta$,
making the choice $q=4d^2/\delta$, in order to have $h_{q-ext}(M_{\mathrm{A}\mathrm{B}})\leq 1-\delta/2$,
yields, after a straightforward computation, exactly the announced exponential decay result.
\end{proof}

The scaling as $(\delta/d)^4$ in the upper bound provided by equation \eqref{eq:sep-bound} of Theorem \ref{th:h_sep-h_qext} is much worse than the scaling as $(\delta/d)^2$ in the upper bound provided by Theorem \ref{th:h_sep-n}. However, equation \eqref{eq:qext-bound} of Theorem \ref{th:h_sep-h_qext}, which relates $h_{sep}(M^{\otimes n})$ to $h_{q-ext}(M)$, may be of interest in some specific cases, namely when $M$ has a maximum overlap with $q$-extendible states which is already of the same order as its maximum overlap with separable states for $q \ll d^2$.

\begin{remark}
By Remark \ref{remark:E_I}, we see that we could also have done the recursive construction described in the proof of Theorem \ref{th:h_sep-h_qext} by imposing instead that, for each $0\leq k\leq n-1$,
\[ p_{k+1} = p_k\Tr_{\mathrm{A}_{i_{k+1}}\mathrm{B}_{i_{k+1}}}\left(M_{\mathrm{A}_{i_{k+1}}\mathrm{B}_{i_{k+1}}} \tau(k)_{\mathrm{A}_{i_{k+1}}\mathrm{B}_{i_{k+1}}}\right),\ \text{with}\ E_I\left(\tau(k)_{\mathrm{A}_{i_{k+1}}\mathrm{B}_{i_{k+1}}}\right) \leq \frac{1}{n-k}\frac{1}{2}\log \frac{1}{p_k}. \]
Now, it is an open question to determine whether there exists a dimension independent constant $C>0$ such that
\begin{equation} \label{eq:E_I-conjecture} E_I\left(\rho_{\mathrm{A}\mathrm{B}}\right)\leq \epsilon\ \Rightarrow\ \exists\ \sigma_{\mathrm{A}\mathrm{B}}\in\mathcal{S}(\mathrm{A}{:}\mathrm{B}) \st \left\| \rho_{\mathrm{A}\mathrm{B}} - \sigma_{\mathrm{A}\mathrm{B}} \right\|_1 \leq C\sqrt{\epsilon}. \end{equation}
If Conjecture \eqref{eq:E_I-conjecture} indeed held, this would imply that the $(p_k)_{1\leq k\leq n}$ satisfy
\[
  \forall\ 0\leq k\leq n-1,\
  p_{k+1} \leq p_k \left(\sqrt{\frac{C^2/2}{n-k}\log \frac{1}{p_k}} + h_{sep}\left(M_{\mathrm{A}\mathrm{B}}\right) \right).
\]
And hence eventually, the following dimension-free exponential decay result for $h_{sep}$:
\[
  h_{sep}(M)\leq 1-\delta\ \Rightarrow\ h_{sep}\left(M^{\otimes n}\right) \leq \left( 1-\frac{\delta^2}{4C^2}\right)^n . \]
And in fact, if a more general variant of Conjecture \eqref{eq:E_I-conjecture} held, with $C\sqrt{\epsilon}$ replaced by $\varphi(\epsilon)$ for $\varphi$ a (universal) non-decreasing function such that $\varphi(0)=0$, then one could prove analogously that
\[ h_{sep}(M)\leq 1-\delta\ \Rightarrow\ h_{sep}\left(M^{\otimes n}\right) \leq \left(1-\frac{\varphi^{-1}(\delta)}{4}\right)^n. \]
The way property \eqref{eq:E_sq-q-ext} of strong faithfulness of squashed entanglement with respect to $q$-extendible states, is proved in \cite{LiW} is relying on the breakthrough result by Fawzi and Renner \cite{FR} that small conditional mutual information does imply approximate recoverability. Now, in an even stronger manner than $E_{sq}(\rho)$ being small means that the conditional mutual information of any extension of $\rho$ is small, $E_I(\rho)$ being small is a condition that is expressible as a bunch of conditional mutual information of extensions of $\rho$ being simultaneously small. So it could be that recoverability results (in particular the best one up-to-date \cite{JRSWW}, which carries the advantage over the original one \cite{FR} of being universal and explicit) would help in an attempt to prove a strong faithfulness property of CEMI with respect to separable states such as \eqref{eq:E_I-conjecture}.
\end{remark}

\begin{theorem}
\label{th:w_sep-h_qext}
Let $M$ be an operator on the tensor product Hilbert space $\mathrm{A}\otimes \mathrm{B}$, satisfying $0 \leq M \leq \Id$. If $h_{sep}(M)\leq 1-\delta$ for some $0<\delta<1$, then for any $n,t\in\N$ with $t\geq(1-\delta+\alpha)n$ for some $0<\alpha\leq\delta$, we have
\[ h_{sep}\left(M^{(t/n)}\right)\leq\left(1-\frac{\alpha^5}{2048\ln 2\, d^4\,(2\delta-\alpha)}\right)^n, \]
assuming $|A|=|B|=d$.
\end{theorem}

\begin{proof}
The proof of this theorem follows a very similar route to that of
Theorem \ref{th:h_sep-h_qext}: For any given state $\alpha_{\mathrm{A}^n}\otimes\beta_{\mathrm{B}^n}$ which is product across the cut $\mathrm{A}^{\otimes n}{:}\mathrm{B}^{\otimes n}$, we want to show that the probability that it passes at least $t$ amongst $n$ tests defined by $M_{\mathrm{A}\mathrm{B}}$ is upper bounded as
\[
  P_t(\alpha_{\mathrm{A}^n}\otimes\beta_{\mathrm{B}^n}) \leq \left(1-\frac{\alpha^5}{2048\ln 2\, d^4\,(2\delta-\alpha)}\right)^n.
\]

In that aim, we start by defining the following deterministic set, number and state: $I_0=\emptyset$, $p_{I_0}=1$ and $\tau(I_0)_{\mathrm{A}^n\mathrm{B}^n}=\alpha_{\mathrm{A}^n}\otimes\beta_{\mathrm{B}^n}$. Then, for each $1\leq k\leq n$, we construct recursively the following random set, number and state: pick $i_k$ uniformly at random in $I_{k-1}^c$, and define
\begin{align*}
& I_k=I_{k-1}\cup\{i_k\}, \\
& p_{I_k}= \Tr_{\mathrm{A}^n\mathrm{B}^n} \left[M^{(I_k)}_{\mathrm{A}^n\mathrm{B}^n} \alpha_{\mathrm{A}^n}\otimes\beta_{\mathrm{B}^n}\right],\\
& \tau(I_k)_{\mathrm{A}_{I_k^c}\mathrm{B}_{I_k^c}}= \frac{1}{p_{I_k}}\Tr_{\mathrm{A}_{I_k}\mathrm{B}_{I_k}} \left[M^{(I_k)}_{\mathrm{A}^n\mathrm{B}^n} \alpha_{\mathrm{A}^n}\otimes\beta_{\mathrm{B}^n}\right].
\end{align*}
Lemma \ref{lemma:E_sq} guarantees that, on average, for each $0\leq k\leq n-1$,
\[
  E_{sq}\left(\overline{\tau}(I_k)_{\mathrm{A}_{i_{k+1}}\mathrm{B}_{i_{k+1}}}\right)
        \leq \frac{1}{n-k}\frac{1}{2}\log \frac{1}{\overline{p}_{I_k}},
\]
so that, on average, for any $q\in\N$,
\[ \overline{p}_{I_{k+1}} \leq \overline{p}_{I_k} \left(\sqrt{\frac{\ln 2\, q^2}{n-k}\log \frac{1}{\overline{p}_{I_k}}} + h_{q-ext}\left(M_{\mathrm{A}\mathrm{B}}\right) \right). \]
In particular, we can make the choice $q=8d^2/\alpha$, in order to have $h_{q-ext}(M_{\mathrm{A}\mathrm{B}})\leq 1-\delta+\alpha/4$. And we thus get from Lemma \ref{lemma:recursivity}, after computation, that on average,
\[ \overline{p}_{I_{k_0}} \leq \left(1-\delta+\frac{\alpha}{2}\right)^{k_0},\ \text{where}\ k_0= \frac{\alpha^4}{1024\ln 2\,d^4\,\log[1/(1-\delta+\alpha/2)]+\alpha^4}(n+1)\geq \frac{\alpha^4}{1024\ln 2\,d^4\,(2\delta-\alpha)}\,n. \]

To finish off the proof, we just have to observe (Cf. \cite{Rao}, Section 8) that
\[ P_t(\alpha_{\mathrm{A}^n}\otimes\beta_{\mathrm{B}^n}) \leq \sum_{I_{k_0}\subset[n],\,|I_{k_0}|=k_0}\frac{1}{{(1-\delta+\alpha)n \choose k_0}}p_{I_{k_0}} \leq \frac{{n \choose k_0}}{{(1-\delta+\alpha)n \choose k_0}} \overline{p}_{I_{k_0}}\leq \left(\frac{n-k_0+1}{(1-\delta+\alpha)n-k_0+1}\right)^{k_0}\left(1-\delta+\frac{\alpha}{2}\right)^{k_0}, \]
where the last inequality follows from the fact that $\prod_{i=0}^{l-1}(a+i)/(b+i)\leq (a/b)^l$, combined with the upper bound on $\overline{p}_{I_{k_0}}$. In the end, we can therefore conclude that
\[ P_t(\alpha_{\mathrm{A}^n}\otimes\beta_{\mathrm{B}^n}) \leq \left(1-\frac{\alpha}{2}\right)^{k_0}\leq \left(1-\frac{\alpha^5}{2048\ln 2\, d^4\,(2\delta-\alpha)}\right)^n, \]
where the first inequality follows from the fact that $(1-\delta+\alpha/2)/(1-\delta+\alpha)\leq 1-\alpha/2$, while the second inequality is a consequence of the lower bound on $k_0$.
\end{proof}

\begin{remark}
Here again, we see by Remark \ref{remark:E_I} that, if
Conjecture \eqref{eq:E_I-conjecture} held, then we could have
obtained in the proof of Theorem \ref{th:w_sep-h_qext} that, on average
\[
  \overline{p}_{I_{k_0}} \leq \left(1-\delta+\frac{\alpha}{2}\right)^{k_0},\
  \text{where}\
  k_0= \frac{\alpha^2}{8\log[1/(1-\delta+\alpha/2)]+\alpha^2}(n+1)\geq \frac{\alpha^2}{8(2\delta-\alpha)}\,n.
\]
And hence eventually, the following dimension-free concentration result for $h_{sep}$:
\[
  h_{sep}(M)\leq 1-\delta\ \Rightarrow\ \forall\ 0<\alpha<\delta,\ \forall\ t\geq(1-\delta+\alpha)n,\
  h_{sep}\left(M^{(t/n)}\right) \leq \left( 1-\frac{\alpha^3}{16C^2(2\delta-\alpha)}\right)^n .
\]
\end{remark}

\section{Equivalence between weak multiplicativity of support functions and of maximum fidelities}
\label{sec:general}

In the previous Sections \ref{sec:sep1} and \ref{sec:sep2}, we studied in great depth one particular example of convex constraint on quantum states, namely the separability one. We showed in this specific case that there is a strong connection between the (weakly) multiplicative behaviour under tensoring of either the support function $h_{sep}$ or the maximum fidelity $F(\cdot,\mathcal{S})$. We would now like to describe, more generally, which kind of convex sets of states exhibit a similar feature.

So let us fix $d\in\N$, $\cH$ a $d$-dimensional Hilbert space, and assume that we have a sequence of convex sets
of states $\cK^{(n)}$ on $\cH^{\otimes n}$, $n\in\N$, with
the following stability properties (under permutation and partial trace):
%\begin{equation} \label{eq:tens-stability}
%\forall\ n\in\N,\ \sigma \in \cK^{(1)}\ \Rightarrow\ \sigma^{\otimes n} \in \cK^{(n)}
%\end{equation}
\begin{equation}
  \label{eq:sym-stability}
  \rho\in\cK^{(n)}\ \Rightarrow\ \forall\ \pi\in\mathcal{S}_n,\ U_{\pi}\rho U_{\pi}^{\dagger}\in\cK^{(n)}\ \text{and}\
  \Tr_{\cH}\rho\in\cK^{(n-1)}.
\end{equation}
Note that requirement \eqref{eq:sym-stability} implies in particular that,
if $\rho^{\otimes n}\in\cK^{(n)}$, then $\rho\in\cK^{(1)}$. In view of our subsequent discussion, it would be meaningless not to impose that the opposite holds as well, i.e.~that, if $\rho\in\cK^{(1)}$, then $\rho^{\otimes n}\in\cK^{(n)}$. This means in other words that, for each $n\in\N$, $\cK^{(n)}$ is assumed to contain the so-called $n^{\text{th}}$ projective tensor power of $\cK^{(1)}$, which is defined as
\[ \left(\cK^{(1)}\right)^{\hat{\otimes} n} := \mathrm{conv}\left\{\rho_1\otimes\cdots\otimes\rho_n,\ \rho_1,\ldots,\rho_n\in\cK^{(1)} \right\}. \]

\subsection{Exponential decay and concentration of $h_{\mathcal{K}}$ from multiplicativity of $F(\cdot,\mathcal{K})$}

Given an operator $M$ on $\cH$, satisfying $0\leq M\leq\Id$, define the support function of $\cK^{(n)}$ at $M^{\otimes n}$ as
\[ h_{\cK^{(n)}}\left(M^{\otimes n}\right) = \underset{\sigma\in\cK^{(n)}}{\sup}\Tr\left(M^{\otimes n}\sigma\right). \]
Define also more generally, for any $0\leq t\leq n$, $h_{\cK^{(n)}}\left(M^{(t/n)}\right)$ as the maximum probability for a state in $\cK^{(n)}$ to pass a fraction $t/n$ of $n$ binary tests $(M,\Id-M)$ performed in parallel. The question we are next interested in is to understand how $h_{\cK^{(n)}}(M^{\otimes n})$ and $h_{\cK^{(n)}}\left(M^{(t/n)}\right)$ relate to $h_{\cK^{(1)}}(M)$.

Hence, assume also that these sets $\cK^{(n)}$ satisfy the following condition: there exists a non-decreasing function $f:\epsilon\in]0,1[\mapsto f(\epsilon)\in]0,1[$ such that, for any state $\rho$ on $\C^d$ and any $0<\epsilon<1$,
\begin{equation} \label{eq:exp-decay-F} \left\|\rho-\cK^{(1)}\right\|_2\geq\epsilon\ \Rightarrow\ F\left(\rho^{\otimes n},\cK^{(n)} \right)^2\leq\left(1-f(\epsilon)\right)^n. \end{equation}

Then, under assumption \eqref{eq:exp-decay-F} for the sets $\cK^{(n)}$, the following holds: for any operator $M$ on $\cH$, satisfying $0\leq M\leq\Id$, and with $\|M\|_2=r$ for some $0\leq r\leq \sqrt{d}$, and any $0<\alpha\leq\delta<1$,
\[ h_{\cK^{(1)}}(M)\leq 1-\delta\ \Rightarrow\ h_{\cK^{(n)}}(M^{\otimes n})\leq \left(1-g(\delta,r)\right)^n\ \text{and}\ \forall\ t\geq (1-\delta+\alpha)n,\ h_{\cK^{(n)}}\left(M^{(t/n)}\right)\leq e^{-ng'(\alpha,r)}, \]
where $g(\delta,r)=f(\epsilon(\delta,r))$ for $0<\epsilon(\delta,r)<1$ the solution of the equation $f(\epsilon)=\delta-r\epsilon$ and $g'(\alpha,r)=f(\epsilon(\alpha,r))$ for $0<\epsilon(\alpha,r)<1$ the solution of the equation $f(\epsilon)=2(\alpha-r\epsilon)^2$.

To come to these statements, the strategy is entirely analogous to the one adopted in the proofs of Theorems \ref{th:h_sep-n} and \ref{th:h_sep-t/n}. It is therefore only sketched below. First of all, when looking for a state $\rho\in\cK^{(n)}$ maximizing $\Tr\left(M^{\otimes n}\rho\right)$, one can in fact assume without loss of generality that $\rho$ is $n$-symmetric. And for such state $\rho$, reasoning as in the proof of Theorem \ref{th:h_sep-n}, we know that we have
\[ \Tr\left(M^{\otimes n}\rho\right) \leq (n+1)^{3d^2} \int_{\sigma\in\mathcal{D}(\mathcal{H})} F\left(\rho,\sigma^{\otimes n}\right)^2 \Tr\left(M\sigma\right)^n \mathrm{d}\mu(\sigma). \]
Hence, we get as a consequence of hypothesis \eqref{eq:exp-decay-F} that, for any $0<\epsilon<1$,
\[ \Tr\left(M^{\otimes n}\rho\right) \leq (n+1)^{3d^2} \left( \left(1-\delta+r\epsilon\right)^n + \left(1-f(\epsilon)\right)^n\right). \]
So choosing $\epsilon$ such that $f(\epsilon)=\delta-r\epsilon$ and $\rho$ such that $\Tr\left(M^{\otimes n}\rho\right)= h_{\cK^{(n)}}(M^{\otimes n})$ yields in particular
\[ h_{\cK^{(n)}}(M^{\otimes n}) \leq 2(n+1)^{3d^2}\left(1-g(\delta,r)\right)^n. \]
Similarly, it follows from hypothesis \eqref{eq:exp-decay-F} as well that, for any $0<\epsilon<1$,
\[ h_{\cK^{(n)}}\left(M^{(t/n)}\right)\leq (n+1)^{3d^2}\left( \exp\left[-2n(\alpha-r\epsilon)^2\right] + \exp\left[-nf(\epsilon)\right] \right), \]
so that choosing $\epsilon$ such that $f(\epsilon)=2(\alpha-r\epsilon)^2$ gives
\[ h_{\cK^{(n)}}\left(M^{(t/n)}\right)\leq 2(n+1)^{3d^2} e^{-ng'(\alpha,r)}. \]
In both cases the polynomial pre-factor $2(n+1)^{3d^2}$ can then be removed by the exact same argument as in the proof of Theorem \ref{th:h_sep-n}.

\subsection{Weak multiplicativity of $F(\cdot,\mathcal{K})$ from exponential decay and concentration of $h_{\mathcal{K}}$}

We would now like to go in the other direction. Namely, let us assume this time that these sets $\cK^{(n)}$ satisfy the following condition: there exists a function $f:(\alpha,d)\in]0,1[\times\N\mapsto f(\alpha,d)\in]0,1[$, non-decreasing in $\alpha$ and non-increasing in $d$, such that, for any operator $M$ on $\cH$, satisfying $0\leq M\leq\Id$, and any $0<\alpha\leq\delta<1$,
\begin{equation} \label{eq:exp-decay-h} h_{\cK^{(1)}}(M)\leq 1-\delta\ \Rightarrow\ \forall\ t\geq (1-\delta+\alpha)n,\ h_{\cK^{(n)}}\left(M^{(t/n)}\right)\leq e^{-nf(\alpha,d)}. \end{equation}

Then, under assumption \eqref{eq:exp-decay-h} for the sets $\cK^{(n)}$, the following holds: for any state $\rho$ on $\cH$ and any $0<\epsilon<1$,
\[ \frac{1}{2}\left\|\rho-\cK^{(1)}\right\|_1 \geq \epsilon\ \Rightarrow\ F\left(\rho^{\otimes n},\cK^{(n)}\right) \leq 2e^{-ng(\epsilon,d)}, \]
where $g(\epsilon,d)=f(\alpha(\epsilon,d),d)/2$ for $0<\alpha(\epsilon,d)<1$ the solution of the equation $f(\alpha,d)/2=(\alpha-\epsilon)^2$.

Here is the strategy to derive such result: Imagine you are given a state on $\cH^{\otimes n}$, which you know is either $\rho^{\otimes n}$ or in $\cK^{(n)}$, and you want to decide between these two hypotheses. For that, you can design a binary test $(T_+,T_-)$ such that outcome $+$ is obtained with a high probability $p$ if the state was $\rho^{\otimes n}$ and outcome $-$ is obtained with a high probability $q$ if the state was in $\cK^{(n)}$. Then, clearly
\[ F\left(\rho^{\otimes n},\cK^{(n)}\right) \leq F\left((p,1-p),(q,1-q)\right) \leq \sqrt{1-p}+\sqrt{1-q}. \]
Therefore, if both error probabilities $1-p$ and $1-q$ are exponentially small, the conclusion follows.

In the present case, the fact that $\left\|\rho-\cK^{(1)}\right\|_1=2\epsilon$, implies that there exist $0\leq M\leq \Id$ and $\epsilon<\eta<1$ such that $\Tr(M\rho)=1-\eta+\epsilon$ whereas $h_{\cK^{(1)}}(M)= 1-\eta$. So consider the binary POVM $(M_0,M_1)=(M,\Id-M)$ performed $n$ times in parallel, and the corresponding binary test $(T_+,T_-)$ with $+$ being the event ``outcome $0$ is obtained more than $(1-\eta+\alpha)n$ times'' and $-$ being the event ``outcome $0$ is obtained less than $(1-\eta+\alpha)n$ times'', for some $0<\alpha<\epsilon$ to be chosen later. Define next, for each $1\leq i\leq n$, the random variable $X_i$, respectively $Y_i$, as the outcome of measurement number $i$ given that the state was $\rho^{\otimes n}$, respectively in $\cK^{(n)}$. Then,
\[ 1-p = \P\left(-\big|\rho^{\otimes n}\right) = \P\left(\sum_{i=1}^n X_i < (1-\eta+\alpha)n\right)\ \text{and}\ 1-q = \P\left(+\big|\cK^{(n)}\right) = \P\left(\sum_{i=1}^n Y_i > (1-\eta+\alpha)n\right). \]
Yet on the one hand, $X_1,\ldots,X_n$ are independent Bernoulli random variables with expectation $1-\eta+\epsilon$, so by Hoeffding's inequality
\[ \P\left(\sum_{i=1}^n X_i < (1-\eta+\alpha)n\right) \leq e^{-2n(\epsilon-\alpha)^2}. \]
While on the other hand, for any $0\leq t\leq n$, $\P\left(\sum_{i=1}^n Y_i > t\right) = h_{\cK^{(n)}}\left(M^{(t/n)}\right)$, so assumption \eqref{eq:exp-decay-h} guarantees that
\[ \P\left(\sum_{i=1}^n Y_i > (1-\eta+\alpha)n\right) \leq e^{-nf(\alpha,d)}. \]
Hence, putting everything together, we eventually obtain that, for any $0<\alpha<\epsilon$,
\[ F\left(\rho^{\otimes n},\cK^{(n)}\right) \leq e^{-n(\epsilon-\alpha)^2} + e^{-nf(\alpha,d)/2}, \]
which yields the wanted result after choosing $\alpha$ such that $f(\alpha,d)/2=(\epsilon-\alpha)^2$.

\begin{remark}
Note that requirement \eqref{eq:sym-stability} is clearly fulfilled by the sets $\mathcal{S}_{\mathrm{A}^n:\mathrm{B}^n}$ of biseparable states on $\left(\mathrm{A}\otimes\mathrm{B}\right)^{\otimes n}$. Furthermore, they satisfy requirements \eqref{eq:exp-decay-F} and \eqref{eq:exp-decay-h} as well, with $f(\epsilon)=\epsilon^2/4$ and $f(\alpha,d^2)=\alpha^2/5d^2$.

It may also be worth emphasizing that conditions \eqref{eq:exp-decay-F} and \eqref{eq:exp-decay-h} are just strengthened and quantitative versions of the following stability property for the sets $\cK^{(n)}$: $\rho\notin\cK^{(1)}\ \Rightarrow\ \rho^{\otimes n}\notin\cK^{(n)}$, i.e.~ equivalently $\rho^{\otimes n}\in\cK^{(n)}\ \Rightarrow\ \rho\in\cK^{(1)}$.
\end{remark}

\subsection{One simple example}

Let us look at what the previous discussion becomes in the case of the simplest possible sequence $\{\mathcal{K}^{(n)},\ n\in\N\}$ satisfying requirement \eqref{eq:sym-stability}, namely when there exists a set of states $\mathcal{K}$ on $\cH$ such that, for each $n\in\N$, $\mathcal{K}^{(n)}$ is exactly the $n^{\text{th}}$ projective tensor power of $\mathcal{K}$, i.e.~\[ \mathcal{K}^{(n)}= \mathcal{K}^{\hat{\otimes} n} := \mathrm{conv}\left\{\rho_1\otimes\cdots\otimes\rho_n,\ \rho_1,\ldots,\rho_n\in\mathcal{K} \right\}. \]
Then, assumption \eqref{eq:exp-decay-h} is clearly satisfied, in the following way: for any operator $M$ on $\cH$, satisfying $0\leq M\leq \Id$, and any $0<\alpha\leq\delta<1$,
\[ h_{\cK^{(1)}}(M)\leq 1-\delta\ \Rightarrow\ \forall\ t\geq (1-\delta+\alpha)n,\ h_{\cK^{(n)}}\left(M^{(t/n)}\right)\leq e^{-n2\alpha^2}. \]
This is a consequence of Hoeffding's inequality, following an argument similar to the one detailed in the previous subsection. And by the result established in the latter, this implies that: for any state $\rho$ on $\cH$ and any $0<\epsilon<1$,
\[ F\left(\rho,\cK^{(1)}\right)\leq e^{-\epsilon}\ \Rightarrow\ F\left(\rho^{\otimes n},\cK^{(n)}\right) \leq 2e^{-n\epsilon^2/8}. \]
This is because $F(\rho,\cK^{(1)})\leq e^{-\epsilon}\ \Rightarrow\ \|\rho-\cK^{(1)}\|_1/2 \geq 1-e^{-\epsilon}\geq \epsilon/2$.

In connection with the discussion developed in Sections \ref{sec:sep1} and \ref{sec:sep2}, we see that we are actually facing the following interesting open question: how differently do $\mathcal{S}(\mathrm{A}^n{:}\mathrm{B}^n)$ and $\mathcal{S}(\mathrm{A}{:}\mathrm{B})^{\hat{\otimes} n}$ behave, from the (more or less equivalent) points of view of support functions and maximum fidelity functions?

\section{De Finetti reductions for infinite-dimensional symmetric quantum systems}
\label{sec:infinite}

All quantum de Finetti theorems and reductions require a bound on the dimension of the involved
Hilbert spaces. So what can be said about symmetric states on $\cH^{\otimes n}$ when $\cH$ is an infinite-dimensional Hilbert space? What extra assumptions do we need on them in order to be able to reduce their study to that of states in some de Finetti form? The original de Finetti reduction of \cite{CKR} was especially designed to prove the security of QKD protocols against general attacks. Yet, showing security of continuous variable QKD is also a major issue. This was the motivation behind the infinite-dimensional de Finetti type theorem of \cite{CR}. Our ultimate goal here is the same, which we rather try to achieve via a de Finetti reduction under constraints.

\subsection{Infinite-dimensional post-selection lemma}

Let $\cH$ be an infinite-dimensional Hilbert space, and let
$\bar{\cH}\subset\cH$ be a (finite) $d$-dimensional subspace of $\cH$.
Denote by $\{\ket{j}\}_{j\in\N}$ an orthonormal basis of $\cH$, chosen such
that $\{\ket{j}\}_{1\leq j\leq d}$ is an orthonormal basis of $\bar{\cH}$.
Then, for any $n,k\in\N$, the $(n+k)$-symmetric subspace of
$\bar{\cH}^{\otimes n}\otimes\cH^{\otimes k}$ is defined as
\[ \Sym^{n+k}\left(\bar{\cH},\cH\right) := \Span\left\{ \sum_{\pi\in\mathcal{S}_{n+k}} \ket{j_{\pi(1)}}\otimes\cdots\otimes\ket{j_{\pi(n+k)}} \st j_1\leq\cdots\leq j_{n+k}\in\N,\ \forall\ 1\leq q\leq n,\ j_q\leq d\right\}. \]
Note that, denoting by $\bar{\cH}_{\perp}\subset\cH$ the orthogonal complement of $\bar{\cH}$, i.e.~$\cH=\bar{\cH}\oplus\bar{\cH}_{\perp}$, we have
\[ \Sym^{n+k}\left(\bar{\cH},\cH\right) \subset \mathrm{V}^{n+k}\left(\bar{\cH},\cH\right) := \bigoplus_{\underset{|I|\geq n}{I\subset[n+k]}} \bar{\cH}_{\perp}^{\otimes I^c}\otimes\Sym\left(\bar{\cH}^{\otimes I}\right). \]

\begin{lemma}
\label{lemma:ps-infinite}
Let $\cH$ be an infinite-dimensional Hilbert space, and let $\bar{\cH}\subset\cH$ be a (finite) $d$-dimensional subspace of $\cH$. Let also $n,k\in\N$. Then, any unit vector $\ket{\theta}\in\Sym^{n+k}\left(\bar{\cH},\cH\right)$ satisfies
\[ \proj{\theta} \leq \left[\sum_{q=0}^k{n+k \choose q} {n+d-1 \choose n}^3 \right] \sum_{\underset{|I|\geq n}{I\subset[n+k]}} \int_{\ket{x}\in S_{\bar{\cH}}} \epsilon(\theta_x)_{\bar{\cH}_{\perp}^{I^c}}\otimes\proj{x}_{\bar{\cH}}^{\otimes I} \mathrm{d}x , \]
where for all $0\leq q\leq k$ and all unit vector $\ket{x}\in \bar{\cH}$, $\epsilon(\theta_x)_{\bar{\cH}_{\perp}^{k-q}}=\Tr_{\bar{\cH}^{ n+q}}\big[\big(\Id_{\bar{\cH}_{\perp}}^{\otimes k-q}\otimes\proj{x}_{\bar{\cH}}^{\otimes n+q}\big)\proj{\theta}\big]$ is a sub-normalized state on $\bar{\cH}_{\perp}^{\otimes k-q}$. %(we mention for the sake of completeness that $\epsilon(\theta_x)_{\bar{\cH}_{\perp}^{k-q}}$ has trace $F\left(\tau(\theta)_{\cH^{ n+q}},\proj{x}_{\bar{\cH}}^{n+q}\right)^2$, with $\tau(\theta)_{\cH^{n+q}}=\Tr_{\cH^{k-q}}\proj{\theta}$ the reduced state of $\proj{\theta}$ on $\cH^{\otimes n+q}$).
\end{lemma}

\begin{proof}
Since $\Sym^{n+k}\left(\bar{\cH},\cH\right)\subset \mathrm{V}^{n+k}\left(\bar{\cH},\cH\right)$, any unit vector $\ket{\theta}\in\Sym^{n+k}\left(\bar{\cH},\cH\right)$ satisfies
\begin{align*}
\proj{\theta} = & P_{\mathrm{V}^{n+k}\left(\bar{\cH},\cH\right)} \proj{\theta} P_{\mathrm{V}^{n+k}\left(\bar{\cH},\cH\right)}^{\dagger}\\
= & {n+d-1 \choose n}^2 \sum_{\underset{|I|,|J|\geq n}{I,J\subset[n+k]}} \int_{\ket{x},\ket{y}\in S_{\bar{\cH}}} \left(\Id_{\bar{\cH}_{\perp}}^{\otimes I^c}\otimes\proj{x}_{\bar{\cH}}^{\otimes I}\right) \proj{\theta} \left(\Id_{\bar{\cH}_{\perp}}^{\otimes J^c}\otimes\proj{y}_{\bar{\cH}}^{\otimes J}\right)^{\dagger} \mathrm{d}x\mathrm{d}y.
\end{align*}
Now, by Lemma \ref{lemma:pinching} (and using the same Caratheodory argument as in the proof of Proposition \ref{prop:ps-pure}), we have
\begin{align*}
& \sum_{\underset{|I|,|J|\geq n}{I,J\subset[n+k]}} \int_{\ket{x},\ket{y}\in S_{\bar{\cH}}} \left(\Id_{\bar{\cH}_{\perp}}^{\otimes I^c}\otimes\proj{x}_{\bar{\cH}}^{\otimes I}\right) \proj{\theta} \left(\Id_{\bar{\cH}_{\perp}}^{\otimes J^c}\otimes\proj{y}_{\bar{\cH}}^{\otimes J}\right)^{\dagger} \mathrm{d}x\mathrm{d}y\\
& \ \ \ \leq \left[\sum_{q=0}^k{n+k \choose q} {n+d-1 \choose n}\right] \sum_{\underset{|I|\geq n}{I\subset[n+k]}} \int_{\ket{x}\in S_{\bar{\cH}}} \left(\Id_{\bar{\cH}_{\perp}}^{\otimes I^c}\otimes\proj{x}_{\bar{\cH}}^{\otimes I}\right) \proj{\theta} \left(\Id_{\bar{\cH}_{\perp}}^{\otimes I^c}\otimes\proj{x}_{\bar{\cH}}^{\otimes I}\right)^{\dagger} \mathrm{d}x.\\
\end{align*}
We then just have to notice that, for any $0\leq q\leq k$ and any unit vector $\ket{x}\in \bar{\cH}$,
\[ \left(\Id_{\bar{\cH}_{\perp}}^{\otimes k-q}\otimes\proj{x}_{\bar{\cH}}^{\otimes n+q}\right) \proj{\theta} \left(\Id_{\bar{\cH}_{\perp}}^{\otimes k-q}\otimes\proj{x}_{\bar{\cH}}^{\otimes n+q}\right)^{\dagger} = \Tr_{\bar{\cH}^{ n+q}}\left[\left(\Id_{\bar{\cH}_{\perp}}^{\otimes k-q}\otimes\proj{x}_{\bar{\cH}}^{\otimes n+q}\right)\proj{\theta}\right]\otimes\proj{x}_{\bar{\cH}}^{\otimes n+q}, \]
in order to actually get the advertised result.
\end{proof}

%\begin{remark} \label{remark:ps}
%In the case where $\ket{\theta}\in\Sym\left(\Sym\left(\cH^{\otimes k}\right)\otimes\Sym\left(\bar{\cH}^{\otimes n}\right)\right)$, Lemma \ref{lemma:ps} can in fact be slightly strengthened into
%\[ \proj{\theta} \leq \begin{pmatrix} n+k \\ k \end{pmatrix} \frac{\begin{pmatrix} n+d-1 \\ n \end{pmatrix}^3}{k!^3} \sum_{\underset{|I|=k}{I\subset[n+k]}}\sum_{\pi\in\mathfrak{S}_k} \int_{\ket{x}\in \bar{\cH}} \epsilon(\theta_{x,\pi})_{\cH^{\otimes I}}\otimes\proj{x}_{\bar{\cH}}^{\otimes I^c} \mathrm{d}x, \]
%where for all $\ket{x}\in \bar{\cH}$ and $\pi\in\mathfrak{S}_k$, $\epsilon(\theta_{x,\pi})_{\cH^{\otimes k}}=\Tr_{\bar{\cH}^{\otimes n}}\left[\left(U(\pi)_{\cH^{\otimes k}}\otimes\proj{x}_{\bar{\cH}}^{\otimes n}\right)\proj{\theta}\left(U(\pi)_{\cH^{\otimes k}}\otimes\proj{x}_{\bar{\cH}}^{\otimes n}\right)^{\dagger}\right]$.

%Indeed, denoting by $P_{s'(k,n)}$ the projector onto $\Sym\left(\Sym\left(\cH^{\otimes k}\right)\otimes\Sym\left(\bar{\cH}^{\otimes n}\right)\right)$, we have
%\[ P_{s'(k,n)} =  \sum_{\underset{|I|=k}{I\subset[n+k]}} P_{\Sym(\cH^{\otimes I})}\otimes P_{\Sym(\bar{\cH}^{\otimes I^c})} = \sum_{\underset{|I|=k}{I\subset[n+k]}}\left(\frac{1}{k!} \sum_{\pi\in\mathfrak{S}_k}U(\pi)_{\cH^{\otimes I}}\right) \otimes \left(\begin{pmatrix} n+d-1 \\ n \end{pmatrix}\int_{\ket{x}\in \bar{\cH}}\proj{x}_{\bar{\cH}}^{\otimes I^c} \mathrm{d}x\right), \]
%and we may then argue in the exact same way as in the proof of Lemma \ref{lemma:ps}.
%\end{remark}

\subsection{An application}

One application of Lemma \ref{lemma:ps-infinite} is to the security analysis of quantum cryptographic schemes, when there is no a priori bound on the dimension of the information carriers. This problem was originally investigated in \cite{CR} via a de Finetti theorem specifically designed for it. It was shown there that, under experimentally verifiable conditions, it is possible to ensure the security of quantum key distribution (QKD) protocols with continuous variables against general attacks. We show here that similar conclusions can be reached using the de Finetti reduction of Lemma \ref{lemma:ps-infinite}.

We look at things from the exact same point of view as the one adopted in \cite{CR}. Let $\cH$ be an infinite-dimensional Hilbert space and let $X,Y$ be two canonical operators on $\cH$. Then, denote by $\Lambda=X^2+Y^2$ the corresponding Hamiltonian, fix $\lambda_0>0$, and define
%\[ \cH_X=\left\{ \ket{\theta}\in\cH \st X^2\ket{\theta}=\lambda\ket{\theta},\ \lambda<\frac{\lambda_0}{2}\right\}\ \text{and}\  \cH_Y=\left\{ \ket{\theta}\in\cH \st Y^2\ket{\theta}=\lambda\ket{\theta},\ \lambda<\frac{\lambda_0}{2}\right\},\]
\[ \bar{\cH}:=\left\{ \ket{\theta}\in\cH \st \Lambda\ket{\theta}=\lambda\ket{\theta},\ \lambda\leq\lambda_0\right\}, \]
finite-dimensional subspace of $\cH$ spanned by the eigenvectors of $\Lambda$ with associated eigenvalue at most $\lambda_0$.

Let $n,k\in\N$, with $n\geq 2k$, and let $\rho^{(n+2k)}$ be a $(n+2k)$-symmetric state on $\cH^{\otimes n+2k}$. Next, define the two events $\mathcal{A}$ and $\mathcal{B}$ as
\begin{align*}
%\mathcal{A}\ = & \ \text{``}\ \forall\ 1\leq q\leq k,\ \left(\frac{1}{2}P_{H_X}+\frac{1}{2}P_{H_Y}\right)\rho_q^{(1)} \left(\frac{1}{2}P_{H_X}+\frac{1}{2}P_{H_Y}\right)^{\dagger}=\rho_q^{(1)}\ \text{''}\\
\mathcal{A}\ = & \ \text{``}\ \forall\ 1\leq q\leq k,\ \Tr\left(\Lambda\rho_q^{(1)}\right) \leq \lambda_0\ \text{''}\\
\mathcal{B}\ = & \ \text{``}\ \exists\ \ket{\theta^{(n+k)}}\in\Sym^{n+k}\left(\bar{\cH}\otimes\bar{\cH}',\cH\otimes\cH'\right):\ \rho^{(n+k)}=\Tr_{\cH'^{n+k}} \proj{\theta^{(n+k)}}\ \text{''},
\end{align*}
where for all $1\leq q\leq k$, $\rho_q^{(1)}=\Tr_{\cH^{n+2k}\setminus\cH_q}\rho^{(n+2k)}$, and $\rho^{(n+k)}=\Tr_{\cH_{k+1}\cdots\cH_{2k}}\rho^{(n+2k)}$. We know from \cite[Lemma III.3]{CR} that there exist universal constants $C_0,c>0$ such that, whenever $\lambda_0\geq C_0\log (n/k)$, we have
\[ \P\left(\mathcal{A}\wedge\neg\mathcal{B}\right) \leq e^{-ck^3/n^2} . \]
In words, this means the following. Fix a threshold $\lambda_0\geq C_0\log (n/k)$, and assume that when measuring the energy $\Lambda$ on the $k$ first subsystems of $\rho^{(n+2k)}$, only values below $\lambda_0$ are obtained. Then, with probability greater than $1-e^{-ck^3/n^2}$, the remaining $n+k$ subsystems of $\rho^{(n+2k)}$ have a purification which is the symmetrization of a state with more than $n$ subsystems supported in $\bar{\cH}\otimes\bar{\cH}'$.

Now, let $\widetilde{\rho}$ be a state on $\cH^{\otimes n+k}$ such that $\widetilde{\rho}=\Tr_{\cH'^{n+k}} \ket{\widetilde{\theta}}\bra{\widetilde{\theta}}$ for some unit vector $\ket{\widetilde{\theta}} \in \Sym^{n+k}\left(\bar{\cH}\otimes\bar{\cH}',\cH\otimes\cH'\right)$. %(see \cite{CR}, Lemma III.5).
Denoting by $d$ the dimension of $\bar{\cH}$, we have by Lemma \ref{lemma:ps-infinite} that $\ket{\widetilde{\theta}}$ satisfies
\[ \proj{\widetilde{\theta}} \leq \left[\sum_{q=0}^k {n+k \choose q}{n+d^2-1 \choose n}^3\right] \sum_{\underset{|I|\geq n}{I\subset[n+k]}} \int_{\ket{x}\in S_{\bar{\cH}\otimes\bar{\cH}'}} \epsilon(\widetilde{\theta}_x)_{(\cH\cH')^{ I^c}}\otimes\proj{x}_{\bar{\cH}\bar{\cH}'}^{\otimes I} \mathrm{d}x. \]
%\[ \ket{\widetilde{\theta}}\bra{\widetilde{\theta}} \leq \begin{pmatrix} n+k \\ k \end{pmatrix} \begin{pmatrix} n+d^2-1 \\ n \end{pmatrix}^3 \sum_{\underset{|I|=k}{I\subset[n+k]}} \int_{\ket{x}\in S_2(\bar{\cH}\otimes\bar{\cH})} F\left(\Id_{\cH\otimes\cH}^{\otimes I}\otimes\proj{x}^{\otimes I^c}, \ket{\widetilde{\theta}}\bra{\widetilde{\theta}}\right)^2 \Id_{\cH\otimes\cH}^{\otimes I}\otimes\proj{x}^{\otimes I^c} \mathrm{d}x, \]
And hence, after partial tracing over $\cH'^{\otimes n+k}$, we finally get
\[ \widetilde{\rho} \leq \left[\sum_{q=0}^k {n+k \choose q}{n+d^2-1 \choose n}^3\right] \sum_{\underset{|I|\geq n}{I\subset[n+k]}} \int_{\ket{x}\in S_{\bar{\cH}\otimes\bar{\cH}'}} \varepsilon(\widetilde{\theta}_x)_{\cH^{I^c}}\otimes\sigma(x)_{\bar{\cH}}^{\otimes I} \mathrm{d}x, \]
where for all $0\leq q\leq k$ and all unit vector $\ket{x}\in \bar{\cH}\otimes\bar{\cH}'$, $\sigma(x)_{\bar{\cH}}=\Tr_{\bar{\cH}'}\proj{x}_{\bar{\cH}\bar{\cH}'}$ is the reduced state of $\proj{x}$ on $\bar{\cH}$, and
$\varepsilon(\widetilde{\theta}_x)_{\cH^{k-q}}=\Tr_{\cH'^{k-q}}\epsilon(\widetilde{\theta}_x)_{(\cH\cH')^{k-q}}$ is the reduced sub-normalized state of $\epsilon(\widetilde{\theta}_x)$ on $\cH^{\otimes k-q}$.

%\begin{remark}
%By Remark \ref{remark:ps}, we actually have the slight refinement
%\[ \widetilde{\rho} \leq \frac{\begin{pmatrix} n+k \\ k \end{pmatrix}\begin{pmatrix} n+d^2-1 \\ n \end{pmatrix}^3}{k!^3} \sum_{\underset{|I|=k}{I\subset[n+k]}}\sum_{\pi\in\mathfrak{S}_k} \int_{\ket{x}\in S_2(\bar{\cH}\otimes\bar{\cH})} \varepsilon(\widetilde{\theta}_{x,\pi})_{\cH^{\otimes I}}\otimes\sigma(x)_{\bar{\cH}}^{\otimes I^c}  \mathrm{d}x. \]
%\end{remark}

Putting everything together, we get the following: Let $n\in\N$ and $k=\lfloor n^{\alpha} \rfloor$ for a given $\alpha$ fulfilling $2/3<\alpha<1$. Let also $\lambda_0$ be a threshold such that on the one hand $\lambda_0\geq C_0\log n$, where $C_0>0$ is a universal constant, and on the other hand $d\leq n^{\beta}$ for a given $\beta$ fulfilling $0<\beta<1/2$. Suppose next that $\rho^{(n+2k)}$ is a $(n+2k)$-symmetric state on $\cH^{\otimes n+2k}$ such that event $\mathcal{A}$ holds. Then, with probability greater than $1-e^{-cn^{3\alpha-2}}$, where $c>0$ is a universal constant, the reduced state $\rho^{(n+k)}$ of $\rho^{(n+2k)}$ on $\cH^{\otimes n+k}$ satisfies
\[ \rho^{(n+k)} \leq (Cn)^{n^{\alpha}+n^{2\beta}} \sum_{\underset{|I|\geq n}{I\subset[n+k]}} \int_{\sigma_{\bar{\cH}}\in\mathcal{D}(\bar{\cH})} \varepsilon(\rho,\sigma)_{\cH^{I^c}}\otimes\sigma_{\bar{\cH}}^{\otimes I} \mathrm{d}\mu(\sigma_{\bar{\cH}}), \]
where $C>0$ is a universal constant, $\mu$ is a probability measure on the set of states on $\bar{\cH}$, and for each $0\leq q\leq k$ and each state $\sigma$ on $\bar{\cH}$, $\varepsilon(\rho,\sigma)_{\cH^{k-q}}$ is a sub-normalized state on $\cH^{\otimes k-q}$.

Now, let $\mathcal{N}:\mathcal{L}(\cH)\rightarrow\mathcal{L}(\mathcal{K})$ be a quantum channel, and assume that there exists some $0<\delta<1$ such that
\[ \sup\left\{ \left\|\mathcal{N}(\sigma)\right\|_1 \st \sigma\in\mathcal{D}(\bar{\cH})\right\} \leq \delta. \]
This implies in particular that, for any $0\leq q\leq k$, and any states $\varepsilon$ on $\cH^{\otimes k-q}$, $\sigma$ on $\bar{\cH}$, we have
\[ \left\|\mathcal{N}^{\otimes n+k}\left(\varepsilon\otimes\sigma^{\otimes n+q}\right)\right\|_1 = \left\|\mathcal{N}^{\otimes k-q}\left(\varepsilon\right)\right\|_1\left\| \mathcal{N}(\sigma)\right\|_1^{n+q} \leq \delta^{n+q}. \]
And subsequently, by what precedes, we obtain the following: For any state $\rho^{(n+2k)}$ on $\cH^{\otimes n+2k}$ such that event $\mathcal{A}$ holds, denoting by $\rho^{(n+k)}$ its reduced state on $\cH^{\otimes n+k}$, we have with probability greater than $1-e^{-cn^{3\alpha-2}}$,
\begin{align*} \left\|\mathcal{N}^{\otimes n+k}\left(\rho^{(n+k)}\right)\right\|_1 \leq & \, (Cn)^{n^{\alpha}+n^{2\beta}} \sum_{q=0}^{k}{n+k \choose q} \sup\left\{ \left\|\mathcal{N}^{\otimes n+k}\left(\varepsilon\otimes\sigma^{\otimes n+ q}\right)\right\|_1 \st \varepsilon\in\mathcal{D}(\cH)^{\otimes k-q},\ \sigma\in\mathcal{D}(\bar{\cH}) \right\}\\
\leq & \, (Cn)^{n^{\alpha}+n^{2\beta}}\sum_{q=0}^{k}{n+k \choose q}\delta^{n+q}\\
\leq & \, (C'n)^{n^{\alpha}+n^{2\beta}}\, \delta^n, \end{align*}
where $C'>0$ is a universal constant. By the way $\alpha,\beta$ have been chosen, this means that, for any $\tilde{\delta}>\delta$ and $n\geq n_{\tilde{\delta}}$, we have with high probability
\[
  \sup\left\{ \left\|\mathcal{N}^{\otimes n+k}\left(\rho^{(n+k)}\right)\right\|_1 \st \rho^{(n+k)}=\Tr_{\cH^{\otimes k}}\rho^{(n+2k)}\ \text{with}\ \rho^{(n+2k)}\in\mathcal{D}(\cH)^{\otimes n+2k}\ \text{such that}\ \mathcal{A}\ \text{holds} \right\}
  \leq \tilde{\delta}^n.
\]

\section{Conclusion and outlook}

We have reviewed (and given a new proof) of the constrained de Finetti
reduction of \cite{DSW}. We have demonstrated its adaptability to various situations
where one would like to impart a known constraint satisfied by a
permutation-invariant state onto the i.i.d.~states occurring in the
operator with which to compare it. We have seen that our technique works especially well in the case of linear constraints (see \cite{DSW} and \cite{LaW2} for two developed such applications).

We have then spent considerable effort on a particularly interesting
convex constraint, separability. Apart from the obvious relevance
to entanglement theory, the constrained de Finetti reduction provides a
very natural framework in which to derive bounds on the success
probability of parallel repetitions of tests, and has immediate
applications in the parallel repetition of $\mathrm{QMA}(2)$, quantum
Merlin-Arthur interactive proof systems with two unentangled provers (see \cite{HM} for further details).
Conversely, we showed that certain progress in entanglement theory
(on the conjectured faithfulness properties of the CEMI entanglement
measure for instance) would imply even stronger, dimension-independent bounds, which
would show in particular that the soundness gap of $\mathrm{QMA}(2)$ can
be amplified exponentially by parallel repetition, without any other
devices. It is curious to see that the progress on questions like this
can depend on the properties of a simple, but little-understood
entanglement measure such as CEMI, and we would like to recommend
its study to the reader's attention. Indeed, it seems to be the best candidate so far for a \emph{magical}, or even \emph{supercalifragilistic}  entanglement measure \cite{Winter}. The latter is defined as one which has the post-selection property with respect to an initial product state and measurement on a separate subsystem (cf.~Lemma \ref{lemma:post-POVM}), is super-additive, and satisfies a universal faithfulness bound with respect to the trace-norm distance (cf.~Conjecture \eqref{eq:E_I-conjecture}).

We have also presented a more abstract framework of convex
constraints, that allows us to demonstrate in greater generality
the interplay between the multiplicative behaviour
of $(i)$ the support function and $(ii)$ the maximum fidelity function. The way $(i)$ is derived from $(ii)$ is via our de Finetti reduction with fidelity weight in the upper bounding operator. And $(ii)$ is obtained from $(i)$ by constructing a test whose failure probability decays exponentially under parallel repetition.

Finally, seeing that the de Finetti reductions had been so far always
limited by the finite dimensionality of the system
involved, we have made first steps towards an extension of
the main technical tool to infinite-dimensional systems under
suitable constraints. It remains to be seen how widely it or a variation can be applied to quantum cryptography in continuous variable systems \cite{CR,CGPLR}, or similar problems.

\section*{Acknowledgments}
%We thank our friends for all the jolly drinks we had together.
This research was supported by the European Research Council (Advanced Grant IRQUAT,
ERC-2010-AdG-267386), the European Commission (STREP RAQUEL, FP7-ICT-2013-C-323970),
the Spanish MINECO (project FIS2013-40627-P), with the support of FEDER
funds, the Generalitat de Catalunya (CIRIT project 2014-SGR-966), and the French CNRS (ANR
projects OSQPI 11-BS01-0008 and Stoq 14-CE25-0033).

\end{document}